\documentclass[final,3p,authoryear,12pt]{elsarticle}
\usepackage[round]{natbib}

\usepackage[colorinlistoftodos,textsize=footnotesize]{todonotes}

\usepackage{amsfonts,amsmath,amssymb,lscape,float,graphicx,tabularx,url,times,theorem,color,natbib,afterpage,rotating,hyperref,xcolor,tikz,moresize,multirow,enumerate,colortbl,mathtools,bm,bbm}
\usepackage[utf8]{inputenc}
\usepackage{mathtools}
\usepackage[caption=true]{subfig}
\usepackage[ruled, vlined]{algorithm2e}
\usepackage{enumitem}

\mathtoolsset{showonlyrefs}
\allowdisplaybreaks 

\usepackage{diagbox}

\newcounter{algsubstate}

\usepackage{setspace} 
 \newcommand{\acknowledgements}[1]{%
  \nonumnote{#1}}

\definecolor{myblue}{rgb}{0.8,0.8,1}
\definecolor{myred}{rgb}{1,0.8,0.8}
\definecolor{mygreen}{rgb}{0.8,1,0.8}
\definecolor{mygrey}{rgb}{220,220,220}

\usepackage[bottom]{footmisc}
\usepackage[scaled]{helvet}

\DeclareMathAlphabet{\xcal}{OMS}{cmsy}{m}{n}

\definecolor{dbblue}{RGB}{10,65,155}				
\definecolor{dbred}{RGB}{215,0,50}					
\definecolor{blue}{RGB}{0,113.9850,188.9550} 			
\definecolor{red}{RGB}{216.7500,82.8750,24.9900} 		
\definecolor{green}{RGB}{118.8300,171.8700,47.9400} 	
\definecolor{grey}{RGB}{110,110,110}				
\definecolor{lgrey}{RGB}{210,210,210}				
\definecolor{c1}{RGB}{0,113.9850,188.9550}			
\definecolor{c2}{RGB}{216.7500,82.8750,24.9900}		
\definecolor{c3}{RGB}{236.8950,176.9700,31.8750}		
\definecolor{c4}{RGB}{125.9700,46.9200,141.7800}		
\definecolor{c5}{RGB}{118.8300,171.8700,47.9400}		
\definecolor{c6}{RGB}{76.7550,189.9750,237.9150}		
\definecolor{c7}{RGB}{161.9250,19.8900,46.9200}		
\definecolor{c14}{RGB}{0,102,102}					

\hypersetup{colorlinks,urlcolor=dbblue,citebordercolor=dbblue,linkbordercolor=dbblue,filecolor=dbblue,filebordercolor=dbblue,citecolor=dbblue,linkcolor=dbblue,colorlinks=true}

\def\E{\mathbb{E}}

\defcitealias{UK7year:2013}{Griffiths et al., 2013}
\usepackage{mathrsfs}

\graphicspath{ {./images/} }

\makeatletter
\def\ps@pprintTitle{%
  \let\@oddhead\@empty
  \let\@evenhead\@empty
  \def\@oddfoot{\reset@font\hfil\thepage\hfil}
  \let\@evenfoot\@oddfoot
}
\makeatother

\usepackage{enumerate}
\usepackage{bbold}
\usepackage{bbm}
\usepackage{subfig}
\usepackage{verbatim,ulem}
\newtheorem{thm}{Theorem}

\newtheorem{prop}{Proposition}

\newtheorem{assume}{Assumption}

\definecolor{cadmiumgreen}{rgb}{0.0, 0.42, 0.24}

\newcommand{\Ac}{\mathcal{A}}

\renewcommand{\E}{\mathbbm{E}}
\newcommand{\F}{\mathbbm{F}}
\newcommand{\N}{\mathbbm{N}}
\renewcommand{\P}{\mathbbm{P}}
\newcommand{\R}{\mathbbm{R}}

\DeclareMathOperator{\arctanh}{arctanh}

\begin{document}

\title{\textit{Forthcoming in Journal of Economic Dynamics and Control}
\\~\\ \textbf{Decentralised Finance and Automated Market Making: Execution and Speculation}}
\author[omi,math]{\'{A}lvaro Cartea}
\author[omi]{Fayçal Drissi}
\author[omi,math]{Marcello Monga}

\address[omi]{Oxford-Man Institute of Quantitative Finance, Oxford, UK}
\address[math]{Mathematical Institute, University of Oxford, Oxford, UK}

\date{November 2023}
\journal{TBC}
\acknowledgements{We are grateful to Tarek Abou Zeid, Philippe Bergault, Patrick Chang, Olivier Guéant, Sebastian Jaimungal, Anthony Ledford, Andre Rzym, and Leandro Sánchez-Betancourt for insightful comments. We are also grateful to participants at the Oxford--ETH Workshop (2022), the BlockSem seminar at École Polytechnique (2022), the Mathematical \& Computational Finance seminar at Oxford (2022), the KCL Financial Mathematics seminar (2022), the $67^\text{th}$  EWGCFM meeting (2023), the Oxford--Princeton Workshop (2024), and the Oxford Victoria Seminar (2024). MM acknowledges financial support from the EPSRC Centre for Doctoral Training in Mathematics of Random Systems: Analysis, Modelling and Simulation (EP/S023925/1). }

\begin{frontmatter}

\begin{abstract}
{Automated market makers (AMMs) are a new prototype of decentralised exchanges which are revolutionising market interactions. The majority of AMMs are constant product markets (CPMs) where exchange rates are set by a trading function. This work studies optimal trading and statistical arbitrage in CPMs where balancing exchange rate risk and execution costs is key. Empirical evidence shows that execution costs are accurately estimated by the convexity of the trading function. These convexity costs are linear in the trade size and are nonlinear in the depth of liquidity and in the exchange rate. We develop models for when exchange rates form in a competing centralised exchange, in a CPM, or in both venues. Finally, we derive computationally efficient strategies that account for stochastic convexity costs and we showcase their out-of-sample performance.}
\end{abstract}


\end{frontmatter}

\section{Introduction \label{sec:Introduction}}

Decentralised Finance (DeFi) is a collective term for blockchain-based financial services that do not rely on intermediaries such as brokers or banks.  New powerful technologies are the engine behind the remarkable growth of DeFi, which is changing the financial landscape and is in direct competition with many traditional stakeholders. Within DeFi,  automated market makers (AMMs) are a new paradigm in the design of decentralised trading venues (DEXs) and are revolutionising the way in which market participants provide and take liquidity. Currently, AMMs are mainly exchanges for cryptocurrencies; however, their core concepts go beyond the cryptocurrency sector and they are poised to challenge traditional centralised exchanges (CEX) in all asset classes.

At present, the majority of AMMs are constant \textbf{function} markets (CFMs), introduced in     \cite{angeris2020improved}. In CFMs, a trading function and a set of rules determine how liquidity takers (LTs) and liquidity providers (LPs) interact, and how markets are cleared.  The trading function is deterministic and known to all market participants.  CFMs display pools of liquidity for pairs of assets, where the exchange rates, i.e., the relative prices between the two assets, are determined by their quantities in the pool, or reserves, as prescribed by the trading function. The trading function establishes the link between liquidity and exchange rates, so LTs can compute the execution costs of their trades as a function of the trade size -- these costs are commonly referred to as slippage or execution costs. A key difference between CFMs and limit order books (LOBs) is that execution costs in CFMs are given by the curvature of the trading function which is known in closed form; see \cite{engel2021presentation} and \cite{angeris2022does}. As in traditional markets that operate an LOB, the larger the size of an order, the higher the execution costs.

Within CFMs, we focus on constant \textbf{product} markets (CPMs), which are the most popular type of CFM and where the trading function uses the product of the reserves to determine clearing rates.  In this paper, we solve the problem of an LT who trades in a CPM to execute a large position in an asset or to execute  statistical arbitrages between the CEX and the DEX.\footnote{Statistical arbitrage is  studied in the context of traditional CEXs and for CEX/DEX arbitrage; see \cite{capponi2023price},  \cite{wang2022cyclic}, \cite{jin2021arbitrage}, \cite{boonpeam2021arbitrage}, \cite{vakhmyanin2023price},  \cite{gogol2024quantifying}, and \cite{li2024pricing}.} We formulate the trading problem as a stochastic control problem in continuous time where the LT controls the speed at which  she sends liquidity taking orders. Key to the performance of the LT's strategies is to balance exchange rate risk and execution costs.  In our model, we use the first-order approximation of the curvature of the trading function to compute execution costs. This approximation is referred to as \textit{convexity costs} and we use Uniswap data to show that 
convexity costs are an accurate estimate of the execution costs studied in \cite{engel2021presentation} and in \cite{angeris2022does}. 
In CPMs, convexity costs are well suited for continuous-time models because they are linear in the trading speed of the LT. Convexity costs are stochastic; specifically, they are inversely proportional to the depth of the pool and proportional to a non-linear transformation of the exchange rate in the pool.  

We use Uniswap data for CPMs that trade pairs of cryptocurrencies to study the empirical properties of this particular DEX. As widely documented in the literature, our analysis shows that rates currently form in the LOBs of alternative electronic CEXs, despite very high levels of activity for many of the pairs traded in Uniswap. Moreover, we find that LP activity is significantly lower than that of LT activity and that current activity and liquidity in AMMs is such that the size of the pool is approximately constant for the trading horizons we consider. Therefore, in our first model, exchange rates form in the CEX. Specifically, we model the lead-follow relationship between the exchange rates in the CPM and those in the CEX, so the LT informs her decisions with the rates in both venues. In this setup, we derive trading strategies to execute large orders and statistical arbitrages. We show that the optimal strategy consists of two components. One component 
unwinds the inventory, and the second 
takes advantage of short-lived discrepancies between the rates in the CPM and those in the CEX. In particular, both components adjust the size of the LT's trades to the stochastic convexity costs in the pool.

In anticipation of the growth of AMMs, and because some assets are  traded only in AMMs, a second version of our model assumes that exchange rates in the DEX are efficient, so CEX/DEX discrepancies are not economically significant. The increase in the efficiency of exchange rates in CFMs will be due to an
increase in the activity of LPs and LTs, which will also result in more changes in the depth of the pool of the CPM. Thus, in our second model, the depth of the pool is stochastic and we show how an LT optimally executes a large position in one of the assets traded in the pool. Finally, in a third  model, rates form  simultaneously in both the CEX and the DEX. In this setup, exchange rates are cointegrated and the LT uses information from both venues to trade optimally. This setup is akin to optimal trading and routing models when an asset is quoted in multiple trading venues; see \cite{frtisch2022economics, angeris2022optimal, henker2024athena}.

In the three models we propose, the optimal strategies are characterised by semilinear partial differential equations (PDEs). The nonlinearity in the PDEs is due to the stochasticity in the convexity costs.  We cannot find closed-form solutions for these equations and solving with numerical techniques is computationally expensive. Thus, we introduce a method that uses piecewise constant convexity costs to approximate the optimal strategy, and we show that the sequence of piecewise-defined strategies converges to a continuous strategy where the trading speed of the LT is determined by the stochastic convexity costs. The LT can deploy the \textit{closed-form approximation strategy} in real time, and we demonstrate that it accurately approximates the optimal strategy in practical scenarios where the misalignment between the CEX and DEX rates is not too large.

We use Uniswap data to illustrate the performance of the liquidation and speculative strategies. The efficient rates are those from Binance where traders interact through a price-time priority LOB. To showcase the performance of our strategy, we use in-sample data to estimate model parameters and  out-of-sample data to execute the strategy in `real time' as an LT would have done. In our analysis, we use rolling-time windows of a few hours between 1 July 2021 and  31 December 2023 to obtain the distribution of the financial performance of the strategies. We look at two pairs of assets, one that is heavily traded and one that is not as frequently traded. We show the superior performance of our liquidation strategy over TWAP and over a strategy that would have executed the whole inventory in one trade at the start of the trading window. In line with \cite{milionis2022quantifying} and \cite{crocswap2022}, our analysis also demonstrates that there are profitable opportunities to execute statistical arbitrages in Uniswap  when the strategy is informed by Binance rates.

Early works on the analysis of the properties of AMMs are in \cite{angeris2019analysis, chiu2019blockchain, lipton2021blockchain, capponi2021adoption}; see also the  literature review in \cite{biais2023advances}. There is a rich literature that studies the economics of AMMs. \cite{lehar2021decentralized} and \cite{capponi2023price}  study rate formation in DEXs, and  \cite{bichuch2022axioms} and \cite{engel2021composing} derive axioms for CFMs. Closer to our work, \cite{engel2021presentation} provide formal derivations for the losses of liquidity providers and execution costs of liquidity takers which we approximate with the convexity costs. On the other hand, \cite{angeris2022does} and \cite{angeris2023replicating} define rate sensitivity and liquidity in AMMs and show that execution costs are closely related to the curvature of the trading function, while \cite{angeris2021CFMM} study liquidity taking in multiple CFMs and \cite{angeris2022optimal} study the optimal routing of liquidity taking orders. 

To the best of our knowledge, this is the first paper to solve optimal execution for LTs in DEXs. Early work on optimal execution in CEXs is in  \cite{bertsimas1998optimal} and \cite{almgren2000optimal}.\footnote{See also \cite{cartea2015book}, \cite{gueant2016book}, \cite{lehalle2018market}, \cite{donnelly2022optimal}, and \cite{webster2023handbook}.} On liquidity provision, \cite{angeris2021analysis} are the first to describe the evolution of the wealth of LPs in CPMs when rates are stochastic. Later, \cite{milionis2022automated} derive a similar loss labeled LVR for CFMs. 
Both works assume rate formation in a CEX and a trading flow which adjusts the reserves in the pool accordingly. In contrast, \cite{cartea2022Predictable} measure the predictable losses of LPs in CFMs and in CL pools with minimal assumptions on the exchange rates and on the trading flow. These losses include the opportunity cost of locking assets in the pool.\footnote{See \cite{cartea2022Predictable} for a detailed comparison.}

The remainder of this paper is organised as follows. Section \ref{sec:AMM} discusses how CFMs operate, and uses Uniswap v3 data to study the dynamics of exchange rates, liquidity, and execution costs in CPMs. Section \ref{sec:Model} solves the optimal trading problem when the depth of the pool is constant throughout the execution window, and rate formation is in an alternative CEX. Section \ref{sec:Model2} considers rate formation in the DEX, and Section \ref{sec:Model3} considers rate formation in both venues. Finally,  Section \ref{sec:Performance} showcases the performance of liquidation and statistical arbitrage strategies, and \ref{sec:apx:dataanalysis} studies LP and LT activity in Uniswap v3 and Binance.

\section{Automated market making \label{sec:AMM}}

In this section, we discuss how CFMs operate and how they differ from electronic markets where traders interact through an LOB. In particular, we describe the interactions of market participants with a CFM that is in charge of a pair of assets.  We use  transaction data from Uniswap v3 to study the activity of market participants, the dynamics of liquidity, and execution costs. 

\subsection{Description \label{sec:AMM1}}

AMMs are hard-coded and immutable programs running on a network. They provide a venue to trade pairs of assets $X$ and $Y$, where the liquidity of the pool consists of $q^X$ units of $X$ and $q^Y$ units of $Y$. The exchange rate of the pool is the price of $Y$ in terms of the price of $X$, and it is determined by the quantities $q^X$ and $q^Y$.\footnote{Some AMMs also display pools with more than two assets.} Two types of market participants interact in an AMM: LPs, who deposit their assets in the pool, and LTs, who trade directly with the pool. 
Here, we consider a CFM in charge of a single pool for the pair of assets $X$ and $Y$.  CFMs are characterised by a deterministic trading function $f(q^X,q^Y)$ that determines the rules of engagement among participants in the pool. For instance, the trading function of the CPM is $f(q^X,q^Y)=q^X\times q^Y$. 

In peer-to-peer networks, participants invoke the code of the AMM smart contract to instruct market operations. LPs send messages with instructions to deposit or withdraw liquidity, and LTs send messages to exchange one asset for the other. To provide liquidity, an LP instructs the AMM with the quantities in assets $X$ and $Y$ to be deposited in a specific pool. On the other hand, LTs indicate to the AMM the pool and the quantity of the asset to be exchanged. The available liquidity in the pool and the trading function of the AMM determine the exchange rate received by the LT. For each trade, LTs pay the AMM a transaction fee, which is distributed amongst LPs in the same proportion as their contributions to the pool.\footnote{See \cite{heimbach2021behavior} and \cite{cartea2022decentralised} for an analysis on how LPs profit from their activity.} 

\subsection{LT and LP trading conditions}

We refer to the rule that governs how LTs trade in the pool by the \textit{LT trading condition}, and to the rule that governs how LPs interact with the pool by the \textit{LP trading condition}.

\paragraph{\textbf{LT trading condition}} The trading function $f\left(q^X,q^Y\right)$ is increasing in $q^X$ and $q^Y,$ and it ties the state of the pool before and after an LT transaction is executed. Throughout, the signs of $x$ and $y$ {indicate the direction of the trade, i.e.,} if $y>0,$ the LT sells asset $Y$, and if $y<0,$ the LT buys asset $Y$. For simplicity, we assume zero fees.\footnote{To include the fee, one applies a discount to the quantity $y$ before calculations are carried out.} The LT trading condition
\begin{equation}\label{eq:trade}
    f(q^X-x,q^Y+y)=f(q^X,q^Y)=\kappa^2\ 
\end{equation}
determines the quantity $x$ that the LT receives (pays) when exchanging $y >0$ ($y <0$). The trading function keeps the quantity $\kappa^2$ constant before and after a trade is executed. We write $f(q^X,q^Y)=\kappa^2$ as $q^X=\varphi(q^Y)$ for an appropriate function $\varphi$ that depends on $\kappa$;  we refer to $\varphi$ as the \textit{level function}. The level function {of  an AMM is convex by design}.\footnote{One can show that a no-arbitrage condition leads to the necessary convexity of the level function; {see  \cite{abernethy2011optimization}, \cite{engel2021composing}, and \cite{cartea2022decentralised}}.}

We denote the (unitary) exchange rate received by the LT when trading a quantity $y$ of asset $Y$ by $Z(y)$, with units $X/Y.$\footnote{The exchange rate $Z(y)$ is the exchange rate received by the LT per unit of asset $Y$, and it is akin to the average execution price for a market order which fills resting limit orders with different price levels in an LOB.} Thus, if an LT wishes to sell a quantity $y$ of asset $Y$, she receives  $x = y \times {Z}(y)$ of asset $X$  in exchange. Therefore, 
\begin{equation*} 
    q^X - x = \varphi(q^Y+ y) \implies
    \varphi(q^Y) - y \, Z(y) = \varphi(q^Y + y) \,,
\end{equation*}
so
\begin{equation}
\label{eq:execrate}
      Z(y) = \frac{\varphi(q^Y)-\varphi(q^Y+y)}{y}\,,
\end{equation}
and  for an infinitesimal quantity $y$ we write
\begin{equation}
\label{eq:marginalrate}
    Z =-\varphi'(q^Y)\, .
\end{equation}

We refer to $Z$  as the \textit{marginal rate} of the AMM, which is equivalent to the midprice in an LOB. The marginal rate $Z$ is a reference exchange rate -- the difference between its value and the execution rate is similar to the difference between the LOB midprice and the average price obtained by a liquidity taking order that crosses the spread and walks the book when it is filled.

For CPMs, the trading function is
\begin{equation}\label{eq:CPMtradingfunction}
f(q^X,q^Y)=q^X\times q^Y=\kappa^2,
\end{equation}
so the level function is $\varphi(q^Y)= \kappa^2 / q^Y\,.$ If an LT trades $y$ when the marginal rate \eqref{eq:marginalrate} is  $Z = -\varphi'(q^Y) = (\kappa/ {q^Y})^2\,,$ then the execution rate \eqref{eq:execrate} is  $\tilde Z(y) = \frac{1}{y} \left(\frac{\kappa^2}{q^Y}-\frac{\kappa^2}{q^Y+y}\right).$

\paragraph{\textbf{LP trading condition}} The trading function $f(q^X,q^Y)$ is increasing in the pool {reserves} $q^X$ and $q^Y.$ Thus, when LP activity increases (decreases) the size of the pool, the value of $\kappa$ increases (decreases). We refer to $\kappa$ as the depth of the pool. A distinctive characteristic of AMMs is that liquidity provision changes the depth of the pool, but it does not change the marginal rate. For example, in a CPM, the marginal rate is the ratio of the quantities supplied in the pool, i.e., $Z=q^X/q^Y\, ,$ and when an LP deposits quantities $x$ and $y$ in the pool, the pair $(x,y)$ must satisfy the LP trading condition
\begin{equation}
\label{eq:LPconditionCPM}
    \frac{q^X}{q^Y} = \frac{q^X+x}{q^Y+y} = Z\, ,
\end{equation}
and the value of $\kappa$ changes from $\sqrt{q^X \times q^Y}\ $ to $\sqrt{\left(q^X+x\right)\left(q^Y + y\right)}$.  For \eqref{eq:LPconditionCPM} to hold, there exists $\rho$ such that $x = \rho \ q^X$ and $y= \rho \ q^Y,$ i.e., liquidity provision and removal by LPs in a CPM is performed in fractions of the pool \textit{reserves} $q^X$ and $q^Y,$ so the depth changes from $\kappa$ to $(1+\rho) \ \kappa.$ 

\subsection{Convexity {costs} and execution costs \label{sec:AMM2_2}}

{In this section, we study the execution costs implied by the liquidity in the pools, see e.g., \cite{engel2021presentation} and \cite{angeris2022does}.\footnote{{\cite{abernethy2011optimization} study execution costs in prediction markets where AMMs have been previously studied.}} We show that these costs are accurately estimated by the \textit{convexity costs}, which are given by the convexity of the trading function. Convexity costs depend on the size of the trade, the depth of liquidity, and the rate in the pool. In particular, they are better suited for continuous-time models because they are linear in the trading speed of the LT. }

The rate in \eqref{eq:execrate} received by an LT deteriorates as the size of the trade increases because the level function $\varphi$ is decreasing, which is {akin to electronic exchanges based on LOBs and other types of trading venues.}\footnote{The trading function $f$ is increasing in $x$ and $y$ and $\partial_y f(x,y) = \partial_y f\left(\varphi(q^Y),y\right) = 0 \text{ so } \varphi'(y) = - \frac{\partial_y f(x,y) }{\partial_x f(x,y)} < 0.$ }   The formulas \eqref{eq:execrate} and \eqref{eq:marginalrate} encode all the information needed by an LT to interact with an AMM. For a trade of size $y$,  the distance between the marginal rate $Z$ and the execution rate $ Z\left(y\right)$ defines execution costs $|y \, \left(Z -  Z(y)\right)|$ in the AMM. We define the \textit{unitary execution costs} as
\begin{equation}
\label{eq:executionCost}
    \textrm{Unitary execution costs} \ = |Z -  Z(y)|\,.
\end{equation}

To further study the characteristics of the execution costs in CPMs and motivate our framework, we use transaction data from Uniswap v3 and the traditional LOB-based exchange Binance. Uniswap v3 is considered the cornerstone of DeFi and is currently the most liquid AMM. We look at the two pairs ETH/USDC and ETH/DAI; see \ref{sec:apx:dataanalysis} for a description of the data. For the pool ETH/USDC, the unit of the depth $\kappa$ is $\sqrt{\textrm{ETH} \cdot \textrm{USDC}},$ of $q^X$ is USDC, of $q^Y$ is ETH, and the marginal rate, the execution rate, and the unitary execution costs are all in $\textrm{USDC} / \textrm{ETH}$; similarly for the pool ETH/DAI. For ease of reading,  we omit the units of $\kappa$ in the remainder of this work.

Here, we analyse the execution costs implied by the pool reserves and the trading function \eqref{eq:CPMtradingfunction}.  More precisely, we analyse the geometry of the constant product trading function to understand how unitary execution costs relate to the depth of the pool and the marginal rate $Z$. Figure \ref{fig:geometry} shows the CPM's level function $\varphi$ for $\kappa = \,$2,500,000. Each point on the curve corresponds to values for $q^X$ and $q^Y$ that result in the same pool depth. Point $O$ corresponds to the current pool quantities $q^X$ and $q^Y$. The slope of the tangent at that point gives the current marginal rate $Z = -\varphi'(q^Y).$

A change from $O$ to $A$ in Figure \ref{fig:geometry} is the result of an LT selling 2,500 ETH. A change from $O$ to $B$ is the result of an LT buying 2,500 ETH. The new rates after these transactions are given by the slopes at the new points $A$ and $B$, respectively. When an LT sells $y=$ 2,500 ETH, the unitary execution rate $Z(y) = x  / y$ is given by the slope of the line ($OA$); see \eqref{eq:execrate}. Similarly, the slope of the line ($OB$) gives the unitary rate for buying $y.$ On the other hand, the unitary execution costs are the absolute difference between the slope of the lines and the slope of the tangent at point $O;$ the magnitude of this difference depends on the curvature of $\varphi$ in the neighbourhood of $O;$ see \cite{engel2021composing} and \cite{angeris2022does}. This curvature is proportional to the convexity of the level function and can be approximated by the second-order Taylor polynomial $\tfrac12 \, \varphi ''(q^Y_O)\, y^2.$ A higher degree of convexity, i.e., more curvature around point $O,$ does not change the slope of the tangent at point $O,$ but changes the slopes of the lines $(OA)$ and $(OB).$ The convexity of the CPM's level function is given by
\begin{equation}\label{eqn: convexity varphi prime prime}
    \varphi ''(q^Y) = \frac{2\,\kappa^2}{{q^Y}^3} = \frac{2\,Z^{3 / 2}}{\kappa}\, .
\end{equation}
Therefore, the execution rate obtained for buying or selling $y$ is always less advantageous than the marginal rate $Z$ because the level function $\varphi$ is convex. 

For orders of ``small size''  one can approximate the unitary execution costs in  \eqref{eq:executionCost} with the convexity costs
\begin{equation}
    \label{eq:execcostsApproxCPM}
    | Z-  Z(y)|\approx   \frac{1}{\kappa}\,Z^{3/2}\, |y|\,,
\end{equation}
or equivalently, approximate the execution rate with
\begin{equation}    \label{eq:execratesApproxCPM}
     Z(y) \approx Z - \frac{1}{\kappa}\,Z^{3/2}\, y\,.
\end{equation}
Clearly, as the depth $\kappa$ increases {or the rate $Z$ decreases}, {the reserves in asset $Y$ in the pool increase, so} the convexity costs \eqref{eq:execcostsApproxCPM} are less pronounced. 

{\footnotesize
\begin{figure}
\footnotesize
\parbox{.45\linewidth}{
\centering
\includegraphics[width=0.53\textwidth]{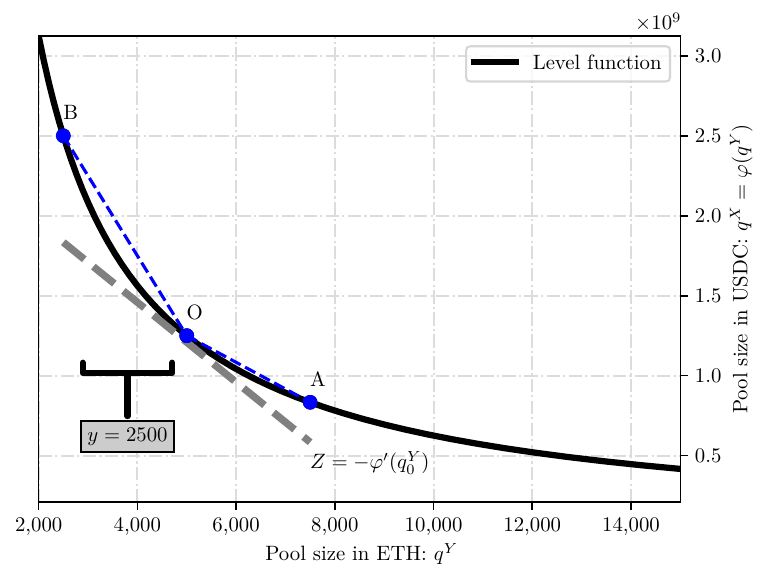}\\
\caption{Geometry of the constant product trading function. The figure shows the function $\varphi\left(y\right) = x$ where $(x,y)$ are the reserves in USDC and ETH.}
\label{fig:geometry}
}
\hfill
\parbox{.45\linewidth}{
\centering
\vspace{2.5mm}
\includegraphics[width=0.55\textwidth]{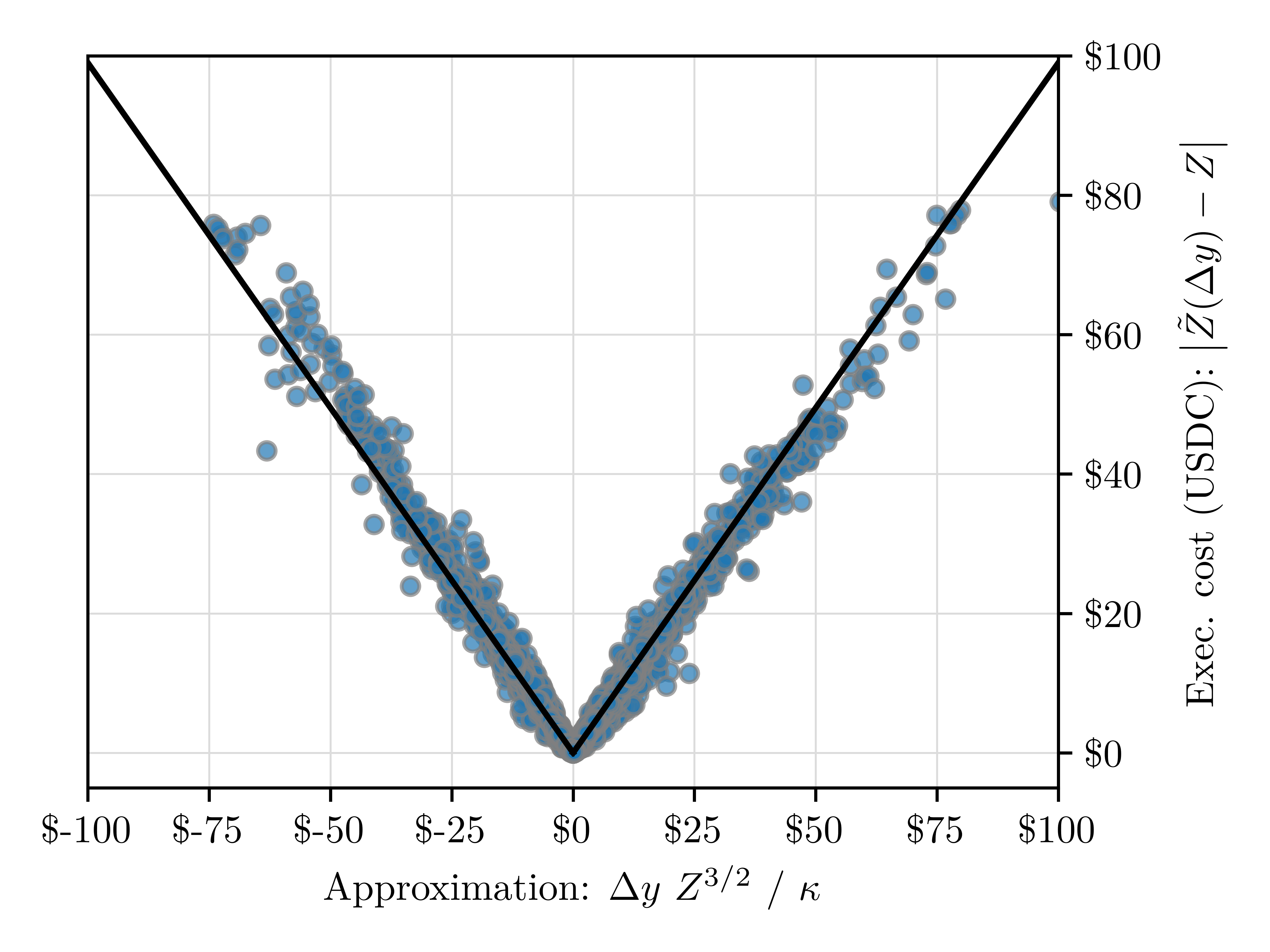}\\ 
\vspace{-1mm}
\caption{Scatter plot of transaction costs and the approximation \eqref{eq:execcostsApproxCPM}
for all transactions in the ETH/USDC pool; see \ref{sec:apx:dataanalysis}.}\label{fig:convexity1}
}
\end{figure}
}

{Next, we employ transaction data from the ETH/USDC pool, see  \ref{sec:apx:dataanalysis}, to compare the convexity costs \eqref{eq:execcostsApproxCPM} with the unitary execution costs \eqref{eq:executionCost} incurred by LTs.} Figure \ref{fig:convexity1} shows a scatter plot of the historical values of the unitary execution costs from transaction data and the approximation $\left(Z^{3/2} /  \kappa\right) y.$ Recall that negative values of $y$ are buy orders and positive values are sell orders. {The figure shows that the convexity costs in  \eqref{eq:execcostsApproxCPM} are an accurate approximation of the unitary execution costs incurred when LTs operate in a CPM.}

In  an optimal trading framework, the LT controls the speed $\nu$ at which she sends orders to the AMM. Now, assume that the LT trades the quantity $y = \nu\, \Delta t$, where  $\Delta t$ is a fixed time-step that determines the LT's frequency of trading and  $\nu$ is fixed during the time interval. The execution rate for $y$ is $Z(y) = Z - Z^{3/2}\,\nu \, \Delta t\,/\,\kappa$. Thus, in continuous time, to reflect the LT's pace of trading we write the LT's execution rate as 
\begin{align}
    \label{eq:execratesApproxCPM_nu}
    Z_t - \frac{\eta}{\kappa}\,Z^{3/2}_t\,\nu_t \, \,,
\end{align}
where the parameter $\eta$ scales the convexity costs according to the LT's trading frequency.

Due to the convexity of the level function, it is sub-optimal to execute large orders in one trade. An optimal trading framework, similar to those developed for traditional LOB-based markets, should balance the trade-off between execution costs and rate risk. The execution cost in CFMs is similar to the cost of ``walking the book'' when trading in LOBs, sometimes referred to as the temporary price impact. The difference between the temporary price impact and the CFM execution cost is that in the CFM we have a deterministic closed-form expression for the execution cost as a function of the depth and the rate, both of which are estimated by LTs to estimate execution costs.   On the other hand, in LOBs, traders usually rely on historical data analysis and assumptions to obtain an estimate of the execution costs. In LOBs, it is generally assumed that temporary price impact is a linear function of the speed of trading where the slope of the function is assumed to be fixed; see \cite{cartea2015book} and \cite{gueant2016book}.

\ref{sec:apx:dataanalysis} investigates trading activity in the pools we consider and compares the dynamics of the exchange rates in Uniswap v3 and Binance. {Our analysis confirms the widely accepted view that despite high levels of activity in current CPMs, exchange rate formation primarily occurs in electronic markets based on LOBs; see \cite{angeris2021analysis}.} In what follows, we call the leading exchange rate from another trading venue the \textit{oracle}, which in our case is the Binance quoted rate. Therefore, in Section \ref{sec:Model}, we propose our first trading model in CPMs where the LT acknowledges that rate formation is not in the pool, so she employs the oracle rate to inform her trading decisions. Moreover, our empirical analysis shows that LP activity is significantly lower than that of LTs, so our model assumes a fixed value for the pool depth $\kappa$ throughout the LT's trading window.

We use the tools of stochastic optimal control to solve the optimisation problem of the LT. The functional form of the convexity costs leads to a semilinear PDE and one can compute the optimal trading speed with a numerical scheme. However, in our case, the numerical scheme is computationally expensive because the semilinear PDE requires a thin grid and a linearisation iterative method to transform the nonlinear problem into a sequence of linear problems. {Thus, to reduce computational costs}, Subsection \ref{sec:cfapprox} introduces a method that uses piecewise constant convexity costs to obtain, {in the limit,} a closed-form approximation strategy {that adjusts the trading speed of the LT according to the stochastic convexity costs in the pool}. We employ this strategy in our performance analysis of Section \ref{sec:Performance}.

{In practice, when  the activity in AMMs is high enough so that rate discrepancies  seldom appear, or when the AMM is the sole trading venue for a pair of assets -- such as many crypto-assets exclusively traded on Uniswap v3 -- our first model is not suitable. Thus, Section \ref{sec:Model2} proposes a model in which exchange rates  form only in the AMM, and where the depth of the pool is stochastic. Later, Section \ref{sec:Model3} proposes a model where both the AMM and the competing CEX are active, so exchange rates form in both markets.

}

\section{Optimal trading {when rates form in the CEX \label{sec:Model}}}

An LT trades in a CPM to exchange a large position in asset $Y$ into asset $X$ or to execute a statistical arbitrage in the pair. In both cases, the LT uses rate information from the pool in the CPM and from another more liquid exchange in which the oracle rate $S$ is the price of $Y$ in terms of that of $X$.  The depth $\kappa$ of the pool is  constant during the execution window $[0,T]$, where $T>0$\,.

The LT must liquidate a position ${y}_0$ in asset $Y$ over the period of time $[0,T]$, and her wealth is valued in terms of asset $X$. The LT trades at speed $(\nu_t)_{t\ge 0}$, so her inventory $(y_t)_{t\ge 0}$ evolves as
\begin{equation} \label{eq:ytildeProcess_modelI}
    d y_t = -\nu_t \,dt\, ,
\end{equation}
where, for simplicity, trading fees are zero.\footnote{{Our performance analysis of Section \ref{sec:Performance} considers  {pool fees and gas fees }to compute the performance of the strategies.}} We do not restrict the speed in \eqref{eq:ytildeProcess_modelI} to be positive; if $\nu>0$, the LT sells the asset, and if $\nu<0$, the LT buys the asset. When the initial inventory is ${y}_0>0$ (resp. ${y}_0=0$) the LT executes a liquidation (resp. speculation) programme. 
{Next, we describe the dynamics of the AMM and oracle rates, and we motivate our modelling assumptions. 

\subsection{Oracle and AMM rate dynamics\label{sec:model:dynamics}}

{In this section, we assume that the oracle rate is the most efficient price, i.e., it incorporates all available information. In practice, for most asset pairs quoted on both a CEX and a DEX, price formation occurs either primarily on the CEX or jointly on both venues. In particular, prior to \textit{The Merge}, when Ethereum transitioned from a Proof-of-Work to a Proof-of-Stake consensus mechanism, high gas fees made small trades economically unviable, so price formation occurred mainly on CEXs. Following \textit{The Merge}, price formation may occur in both markets, in which case the model of Section $5$ is more appropriate.}

The dynamics of the unaffected oracle exchange rate process $(S_t)_{t\ge 0}$  are 
\begin{equation} \label{eq:SProcess}
    dS_t = \sigma\, S_t\, dW_t\,, 
\end{equation} 
where the volatility parameter $\sigma$ is a nonnegative constant, and $(W_t)_{t\ge 0}$ is a standard Brownian motion. On the other hand,  the unaffected marginal rate in the AMM is denoted by $(Z_t)_{t\ge 0}$.  {In the absence of a CEX where prices form, the oracle $S$ is interpreted as a fundamental value observed by the agent. In this case, the main economic implications of our model remain valid.}

In the absence of market frictions, continuous arbitrage between the oracle rate $S$ and the marginal rate $Z$ would make exchange rates converge so that $S_t = Z_t$ at any time $t$. However, exchanges and AMMs are not frictionless, and a portion of the trading flow is not informed, so we consider the following dynamics for the AMM's unaffected marginal rate:
\begin{equation}
\label{eq:ZProcess}
dZ_{t}=\beta\,\left(S_{t}-Z_{t}\right)\,dt+\gamma\,Z_{t}\,dB_{t}\, ,
\end{equation}
where $\beta > 0$ is the oracle-reverting parameter that quantifies the strength of the arbitrage trading flow, $\gamma > 0$ quantifies the dispersion induced by the noise trading flow, and $(B_t)_{t\ge 0}$ is a standard Brownian motion independent of $(W_t)_{t\ge 0}.$\footnote{Similar dynamics to those in  \eqref{eq:ZProcess} are explored in the literature on optimal execution in CEXs that uses market signals to improve performance of strategies; see  \cite{cartea2016incorporating}, \cite{BechlerLudkovski},  \cite{cartea2018enhancing}, \cite{lehalle2019incorporating},  \cite{neuman2020optimal}, \cite{forde2022optimal},  \cite{belak2018optimal}, and \cite{cartea2018trading}.}

{The parameter $\beta$ captures the speed and efficiency with which arbitrageurs align prices across both markets. Specifically, when liquidity is costly in either the CEX or the DEX, the value of $\beta$ tends to be lower. Figure~\ref{fig:stats dynamics} shows estimates of $\beta$ when the CEX is Binance and the DEX is Uniswap v3, for the liquid pair ETH/USDC and the less liquid pair ETH/DAI. As expected, the magnitude of $\beta$ is significantly higher for the more liquid ETH/USDC pair, reflecting more active arbitrage and deeper liquidity in both the CEX and the DEX, which reduces trading frictions for arbitrageurs.}

\paragraph{\textbf{Rate impact of the LT's trades}} Our model considers the execution of large trades in an AMM. Akin to the literature on the execution of large orders in CEXs, we assume that the LT's trading activity has both temporary and permanent impact on the pool rates. Temporary impact is given by the convexity costs  and permanent impact is linear in the LT's speed of trading; see \cite{gatheral2013dynamical}.   Specifically, the pool quoted rate process is $(\tilde Z_t)_{t\ge 0}$  and the oracle market rate is $(\tilde S_t)_{t\ge 0}$ with dynamics
\begin{equation}\label{eq:market rates}
d\tilde{Z}_{t}=dZ_{t}-c\,\nu_t\,dt\,, \qquad d\tilde{S}_{t}=dS_{t}-c\,\nu_t\,dt\,, 
\end{equation}
where $\tilde Z_0 = Z_0$, $ \tilde S_0 = S_0\,,$ and  $c\ge 0$ is the permanent impact parameter. The dynamics of $\tilde Z$ in \eqref{eq:market rates} imply that 
the trading activity of the LT affects the marginal rate and the level to which it reverts. More precisely, write
\begin{align*}
    d\tilde Z_t = dZ_{t}-c\,\nu_{t}\,dt &= \beta\,\left(S_{t}+c\,\left(y_{t}-y_{0}\right)-\tilde{Z}_{t}\right)dt+\gamma dB_{t}-c\,\nu_{t}\,dt\\
    & = \beta\left(\tilde{S}_{t}-\tilde{Z}_{t}\right)dt+\gamma\,dB_{t}-c\,\nu_{t}\,dt\,.
\end{align*}
Thus, the pool's quoted rate $\tilde Z$ reverts to the oracle market rate $\tilde S$.  

{In our model, permanent price impact scales linearly with the agent's trading speed. This specification reflects how net buying (selling) pressure over the trading window induces a permanent upward (downward) price shift in the AMM due to the deterministic pricing rules of the AMM described above. This shift is linked to the depth parameter $\kappa$, which governs how price responds to volume. In this case, our assumption of instantaneous permanent market impact is a tractable approximation that smooths the total price impact over time.}

\paragraph{\textbf{Suitability of rate dynamics}} To assess the suitability of the exchange rate dynamics in \eqref{eq:SProcess} and \eqref{eq:ZProcess}, 
Figure \ref{fig:stats dynamics} shows the results of a CEX/DEX Granger-causality test for different data sampling frequencies.  The results of the  test  indicate that Binance rates lead the rates in the AMM. This is not by design, it is a consequence of the higher liquidity in Binance. Thus, currently, and similar to the dynamics in \eqref{eq:ZProcess}, it is crucial to consider the rate from a more liquid venue when trading in an AMM. Moreover, our analysis in \ref{sec:apx:dataanalysis} shows that there is more trading activity in Binance than there is LT and LP trading activity in the AMM, as measured by the frequency of instructions and the size of the orders. In particular, during periods with little trading activity in the pools,  the oracle rate plays a central role to attract LT  activity in the AMM whenever the difference between the marginal rate and the oracle rate is significant; recall that only liquidity taking trades can change  the rate of the pool. The widening of the difference between the two rates triggers LT activity which drives the two exchange rates to converge; i.e., arbitrageurs keep markets in check.

\begin{figure}%
    \centering
    \subfloat[ETH/DAI $0.3\%$.] {{\includegraphics[width=0.45\textwidth]{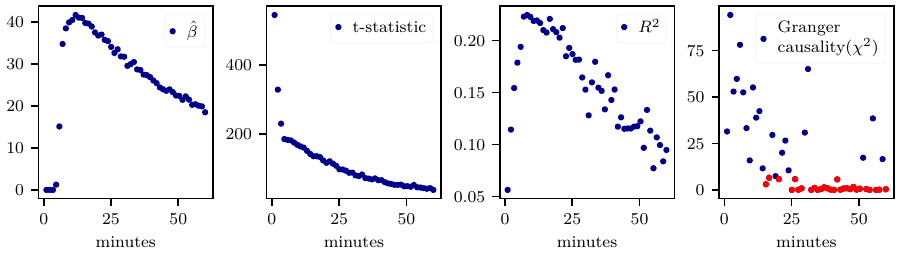} }}%
    \qquad
    \subfloat[ETH/USDC $0.05\%$.] {{\includegraphics[width=0.45\textwidth]{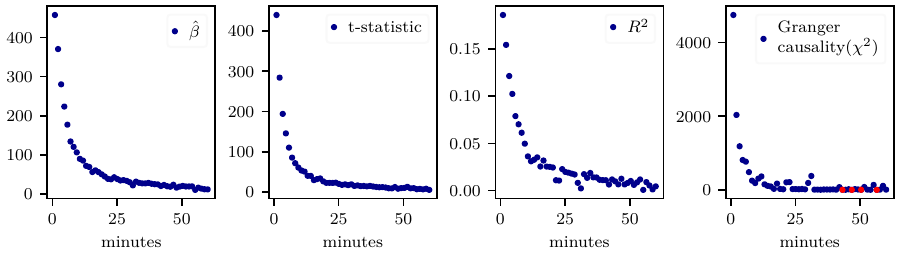} }}%
    \caption{{For each pair of assets and for values of the sampling frequency $\Delta t \in [14 \text{ seconds}, 60 \text{ minutes}]$, the first panel shows the estimated oracle-reversion parameter $\beta$ in \eqref{eq:ZProcess}, the second panel shows the associated $t$-statistic, the third panel shows the R-squared of the regression \eqref{eq: discrete rate dyn 0}, and the last panel shows the Granger causality test statistic $\chi^2$. Red dots in the last panel correspond to rejection, at the $0.01\%$ confidence level, of the hypothesis that CEX rates Granger-cause DEX rates.  }}
    \label{fig:stats dynamics}%
\end{figure}

The dynamics in \eqref{eq:ZProcess} consider a continuous oracle-reversion behaviour of the rate in the AMM. However, for some pairs of assets for which rates form in CEXs, \cite{milionis2023effect} and \cite{he2024liquidity} show that there are no-trade intervals, induced by execution costs in the DEX, within which the DEX rate does not change despite changes in the CEX rate.  To study the suitability of the marginal rate dynamics in these cases, we discretise \eqref{eq:ZProcess} as
\begin{equation}
\label{eq: discrete rate dyn 0}
\Delta \log Z_t  = -\frac{\gamma^2}{2} \, \Delta t + \beta\, \left(\frac{S_t - Z_t}{Z_t}\right) \,\Delta t + \gamma \, \sqrt{\Delta t} \, \epsilon_t\,,
\end{equation}
where  $\epsilon_t$ is an error term, and we employ data for the two pairs ETH/USDC and ETH/DAI to run an ordinary least squares (OLS) regression to estimate the parameters in \eqref{eq:ZProcess}.\footnote{Binance and Uniswap data for the two pairs ETH/USDC and ETH/DAI are described in \ref{sec:apx:dataanalysis}.}

Figure \ref{fig:stats dynamics} shows the estimation results for multiple values of the sampling frequency $\Delta t$.  The figure shows that the descriptive power of the dynamics \eqref{eq:ZProcess} for the marginal rate in the liquid pool ETH/USDC improves as the  frequency of data sampling increases -- this improvement is primarily due to the presence of noise trading and arbitrage trading activity at high frequency.  Thus,  the impact of no-trade zones in liquid pools is negligible. However, for the illiquid pool ETH/DAI and at high sampling frequencies ($\le 5$ minutes) of the data, Figure \ref{fig:stats dynamics} shows that model \eqref{eq:ZProcess} does not describe the marginal rate accurately due to intervals with no trades. However, when data are sampled at low frequencies ($\ge 5$ minutes), the impact of these no-trade intervals diminishes. As the sampling frequency decreases, this reduction is attributed to larger movements in both rates due to more noise trading and arbitrage trading. Thus, optimal trading based on the rate dynamics in  \eqref{eq:ZProcess} is suited for active pools at any observation frequency, and for illiquid pools when the observation frequency is low.

{Finally, the sampling frequency that yields the highest value of $\beta$ in Figure~\ref{fig:stats dynamics} corresponds to the frequency with which arbitrageurs align prices across markets for a given pair. This frequency corresponds to the timescale at which observed CEX/DEX price differences exhibit maximal mean-reversion behavior, leading to higher estimates of the parameter $\beta$. For the liquid pair ETH/USDC in Figure~\ref{fig:stats dynamics}, arbitrage activity occurs at the highest sampling frequency, namely, block creation frequency. In contrast, for the less liquid pair ETH/DAI, arbitrage activity occurs less frequently, approximately every $10$ minutes. More generally, when CEX/DEX prices are sampled at a lower frequency than that of the underlying arbitrage activity, the observed spread between prices appears less mean-reverting, leading to a lower estimated value of $\beta$.}

}

\subsection{Model setup\label{sec:model:optimal}}

As discussed above, when an LT sends a trade to the AMM, the rate impact received by the order is encoded in the trading function. As in traditional models for optimal execution, the  impact depends on the trading speed. In AMMs, rate impact also depends on the rate $\tilde Z$ and  the depth $\kappa$. Specifically, the difference between the execution rate and the marginal rate $\tilde Z$ is as in \eqref{eq:execratesApproxCPM_nu}, {i.e., the convexity costs to trade at speed $\nu$ are $\frac{\eta}{\kappa}\, \tilde Z^{3/2}\,\nu\,.$} The LT trades at speed $\nu$, so the quantity of $Y$ swapped at every instant in time is given by $\nu\,dt$, and the dynamics of her holdings in asset $X$ are given by
\begin{equation}\label{eq:cash}
    d \tilde x_t =\left(\tilde Z_t-\frac{\eta}{\kappa} \, \tilde Z_{t}^{3/2} \,\nu_{t}\right)\,\nu_t\,dt\, .
\end{equation}

The cost term in \eqref{eq:cash} is the rate impact function of the AMM, which is determined by the convexity of the level function $\varphi$ that depends on the {quoted} marginal rate $\tilde Z_t$ of the pool at the time the liquidity taking order is executed, see \eqref{eq:execcostsApproxCPM}. The key difference between the functional form of the convexity costs in \eqref{eq:cash} and those in the equity LOB literature is that in general, the price impact functions proposed for LOB models do not depend on the price of the asset, see \cite{cartea2015book}.  Thus, the convexity costs are stochastic and their dynamics are known.\footnote{{Existing work on dynamic market impact models for CEXs is in  \cite{almgren2012optimal}, \cite{gatheral2013dynamical}, \cite{cheridito2014optimal}, \cite{barger2019optimal}, and \cite{fouque2021optimal}. In particular, \cite{graewe2018smooth} considers price-dependent impact functions.}}

The LT trades at speed $\nu$ to maximise her expected terminal wealth in units of $X$ while penalising inventory in $Y$. Specifically, her performance criterion is given by
\begin{align}
\label{eq:perfcriteria_modelI}
    u^{\nu}\left(t,\tilde x,y,\tilde Z,\tilde S\right)=\E_{t,\tilde x,y,\tilde Z,\tilde S}\left[\tilde x_{T}^{\nu}+y_{T}^{\nu}\,\tilde Z_T-\tilde \alpha\,\left(y_{T}^{\nu}\right)^{2}-\phi\,\int_{t}^{T}\left(y_{s}^{\nu}\right)^{2}\,ds\right]\,.
\end{align}
The first term on the right-hand side of \eqref{eq:perfcriteria_modelI} represents the LT's holdings in asset $X$ at the end of the trading window,  {the second term represents the LT's earnings from liquidating her remaining inventory $y_T$ at the terminal time $T$ at rate $\tilde Z_T$, and the third term is the `cost' of liquidating this final inventory. In particular, in the limit $\alpha\rightarrow\infty,$ the optimal strategy guarantees full liquidation. The terminal penalty $\tilde \alpha$ may incorporate a non-financial component to refine the search for the optimal strategy, which reflects the LT's preference for achieving complete liquidation.} 
Finally, the last term on the right-hand side of \eqref{eq:perfcriteria_modelI} represents a running inventory penalty, {which does not affect the wealth of the LT.} Here, the parameter $\phi \geq 0$ quantifies the urgency of the LT to liquidate inventory. The units of $\phi$ and $\alpha$ are such that the penalty terms are in units of $X$. 

{In our model, the liquidity depth $\kappa$ in the AMM is taken as given by the liquidity taker. While liquidity provision is passive in the current setup, this does not preclude it from being strategic. For example, in a two-stage model, where liquidity providers choose the pool’s depth in the first stage based on expectations about future trading volumes of LTs in the second stage, $\kappa$ would be determined strategically in the first stage. Liquidity takers would then take this depth as given when optimising their trading decisions. In such a framework, the results of our model would still apply in the second stage.
}

In \ref{apx:originalmodel}, we use stochastic optimal control to study the optimisation problem in \eqref{eq:perfcriteria_modelI}. The functional form of the convexity costs leads to {the semilinear PDE \eqref{eq:ANX:hjbtheta2}}  which we cannot  solve in closed form. The optimal trading speed in feedback form is a function of the solution to the semilinear PDE, so one must compute the optimal trading speed using a numerical scheme, which in our case is  computationally expensive.\footnote{{The numerical scheme uses a {four}-dimensional grid and requires iterative methods to linearise the problem. }}   In practice, the profitability of execution and statistical arbitrage strategies relies on computing the strategy and instructing the AMM within very short periods of time.   {Thus, in Subsection \ref{sec:cfapprox},  we introduce a method that uses piecewise constant convexity costs to obtain, in the limit when the number of piecewise constant intervals increases to infinity, a continuous closed-form approximation strategy that accounts for stochastic convexity costs and that can be deployed by the LT in real time.} Finally, to assess the accuracy of the closed-form approximation strategy, Subsection \ref{sec:model:comparison} compares it with the optimal strategy {obtained by numerical solving of the semilinear PDE  \eqref{eq:ANX:hjbtheta2}. }

\subsection{Closed-form approximation strategy \label{sec:cfapprox}}

\paragraph{\textbf{Constant impact parameter strategy}} To obtain a trading strategy that can be implemented by the  LT in real time, we first derive a strategy where  convexity costs are deterministic. We use this  strategy as the building block for the LT's closed-form approximation strategy we derive below. Accordingly, we write the convexity cost $\eta\,Z^{3/2}\,\nu/\kappa$ in \eqref{eq:cash} as    $- \eta\,\zeta\,\nu\,,$ where $\zeta > 0$ is the impact parameter and recall that the value of $\eta$ depends on the LT's trading frequency. With fixed executions costs,  the LT can derive a closed-form optimal trading strategy $(\nu^{\star,\zeta}_t)_{t\ge 0}$ for a given value of the parameter $\zeta$.

Here, the LT trades at speed $(\nu_t^{\zeta})_{t\ge 0}$, so the inventory $(y_t^{\zeta})_{t\ge 0}$ evolves as in \eqref{eq:ytildeProcess_modelI} and the dynamics of the LT's holdings $(\tilde x_t^{\zeta})_{t\ge 0}$ in asset $X$ are 
\begin{equation}\label{eq:impact_deterministic}
    d \tilde x^\zeta_t =\left(\tilde Z_t-\eta \, \zeta \,\nu^\zeta_{t}\right)\,\nu^\zeta_t\,dt\, .
\end{equation}

The performance criterion of the LT, who trades at speed $\nu^{\zeta}$\,, is given by
\begin{equation}\label{eq:perfcriteria0}
    \omega^{\zeta}\left(t,\tilde x,y,\tilde Z,\tilde S\right)=\E_{t,\tilde x,y,\tilde Z,\tilde S}\left[\tilde x_{T}^{\zeta}+ y_{T}^{\zeta}\,\tilde Z_T-\tilde\alpha\,\left(y_{T}^{\zeta} \right)^{2}-\phi\,\int_{t}^{T}\left(y_{s}^{\zeta}\right)^{2}\,ds\right]\,.
\end{equation}

{
In the remainder of this work, we make the  assumption that the terminal penalty $\tilde\alpha$ is larger than half the permanent rate impact.
\begin{assume}\label{ass:alpha large}
$\alpha = \tilde\alpha - c/2 \ge 0$.
\end{assume}
In practice,  Assumption \ref{ass:alpha large} is not restrictive because the value of the permanent rate impact is typically small and agents choose arbitrarily large values of $\tilde \alpha$ to enforce liquidation.  The next result demonstrates that when Assumption \ref{ass:alpha large} holds, the solution to the optimal execution problem \eqref{eq:perfcriteria0} becomes independent of the permanent rate impact; see \ref{apx:proofprop1} for a proof. 
\begin{prop}\label{prop:perm impact is useless}
Let Assumption \ref{ass:alpha large} hold. Problem \eqref{eq:perfcriteria0} is equivalent to 
\begin{align}
\label{eq:perfcriteria}
    u^{\zeta}(t,x,y,Z,S)=\E_{t,x,y,Z,S}\left[x_{T}^{\zeta}+y_{T}^{\zeta}\,Z_T-\alpha\,\left(y_{T}^{\zeta}\right)^{2}-\phi\,\int_{t}^{T}\left(y_{s}^{\zeta}\right)^{2}\,ds\right]\ ,
\end{align}
where $\alpha=\tilde \alpha - c/2$ and the unaffected cash process $(x_t^{\zeta})_{t\ge 0}$ satisfies
\begin{equation}\label{eq:unaffected cash dynamics}
dx^{\zeta}_t = \left(Z_t - \eta\,\zeta\,\nu^\zeta_t\right)\,\nu^\zeta_t\,dt\,, \quad x_0^\zeta = \tilde{x}_0^\zeta\,.
\end{equation}
\end{prop}

}

In {\ref{apx:standard model}, we follow similar steps as in the classical models of \cite{cartea2015book} and \cite{gueant2016book} to derive the optimal strategy of  problem  \eqref{eq:perfcriteria}.}

\paragraph{\textbf{The closed-form approximation strategy}} Here,  we use a family of closed-form strategies of the type \eqref{eq:perfcriteria} to derive a piecewise-defined trading strategy which approximates the optimal trading speed that maximises the performance criterion \eqref{eq:perfcriteria_modelI}. Specifically, we partition the space of the rate $Z$ into strips and define a piecewise strategy which uses a different impact parameter $\zeta$ in each different strip. Finally, we show that as the width of the strip becomes arbitrarily small, the piecewise strategy converges to the closed-form approximation strategy. 

Let $\left\{Z_0, \dots,Z_{N}\right\}$ be a partition of $\left[\underline{Z},\overline{Z}\right],$ where $0<\underline{Z}<\overline{Z}$ so that for each $N\in\N$ and  $j\in\{0,\dots,N\}$ we define
\begin{equation}\label{eq:partitionedConvexity}
    Z_{j}^N \coloneqq \underline{Z}+\frac{j}{N} \left(\overline{Z}-\underline{Z}\right)\quad\text{and}\quad\zeta_{j}^{N} =  \frac{1}{\kappa}\, \left(Z_{j}^N\right)^{3/2}\,.
\end{equation}
In the remainder of this section, $\nu^{\star,j,N}$ denotes the optimal trading strategy with impact parameter $\zeta_{j}^{N}$. Note that whenever $Z$ is arbitrarily close to $Z_{j}^N$, the impact parameter  $Z^{3/2}/\kappa$ in \eqref{eq:cash} can be approximated by $\zeta_{j}^{N}$. Thus, to construct the approximate trading strategy, we first define a strategy $\nu^*_N$ that uses the closed-form optimal trading speed $\nu^{\star,j,N}$ to approximate the optimal trading speed whenever the rate is close to $\zeta_{j}^{N}$. We define the piecewise-defined trading speed  $\nu^*_N$ 
\begin{equation}\label{eq:piecewise_strategy}
    \begin{split}
        \nu^{\star,N}\left(t,{y},Z,S\right) =\,& \nu^{\star,0,N}\left(t,{y},Z,S\right)\mathbbm{1}_{Z<Z^N_1} + \sum_{j=1}^{N-1}\nu^{\star,j,N}\left(t,{y},Z,S\right)\mathbbm{1}_{Z\in[Z_j^N,Z^N_{j+1})} \\
        &+ \nu^{\star,N,N}\left(t,{y},Z,S\right)\mathbbm{1}_{Z\geq Z^N_N}\,.
    \end{split}
\end{equation}
The strategy $\nu^{\star,N}\left(t,{y},Z,S\right)$ has first-type discontinuity points; specifically, it is discontinuous over $[0,T]\times\R\times\{Z_j^N\}\times\R\,$ for each $j\in\{1,\dots,N\}$ because for each $\left(t,{y},Z^N_{j+1},S\right)\in[0,T]\times\R^2$ we have  $\nu^{\star,j,N}\left(t,{y},Z^N_{j+1},S\right)\neq\nu^{\star,j+1,N}\left(t,{y},Z^N_{j+1},S\right)$.

The theorem below shows how to partition  $\left[\underline{Z},\overline{Z}\right]$ to make the size of the discontinuities in $\nu^{\star,N}\left(t,{y},Z,S\right)$ arbitrarily small. Furthermore, when the distance between points in the partition becomes sufficiently small, the sequence of piecewise-defined optimal strategies $\{\nu^{\star,N}\}_{N\in\N}$ converges uniformly to a continuous closed-form approximation strategy which we use in our performance study of Section \ref{sec:Performance}. 

\begin{thm}\label{thm}
For each $\varepsilon>0\,$ there exists $N\in\N$ such that
\begin{equation}\label{eq:inequality}
    \max_{j=1,\dots,N}\left|\nu^{\star,j,N}\left(t,{y},Z^N_{j+1},S\right)-\nu^{\star,j+1,N}\left(t,{y},Z^N_{j+1},S\right)\right|<\varepsilon\,.
\end{equation}
Furthermore, for each $N\in\N\,,$ let $\hat{\nu}^{\star,N}\coloneqq \nu^{\star,N}|_{[0,T]\times\R\times[\underline{Z},\overline{Z}]\times\R}$. Then, the sequence $\{\hat{\nu}^{\star,N}\}$ converges to $\hat{\nu}^{\star}$ uniformly in $ [0,T]\times\R\times\left[\underline{Z},\overline{Z}\right]\times\R $\,, where
\begin{equation}
\label{eq:pseudooptimal}
    \hat{\nu}^{\star}\left(t,{y},Z,S\right)=-\frac{\kappa}{\eta}\,Z^{-3/2} \,A(t,Z)\, y+\frac{\kappa}{2\,\eta}\,Z^{-3/2}\,B(t,Z)\,(S-Z)\,,
\end{equation}
and
\begin{equation}\label{eq:pseudoAB}
    \begin{split}
        A(t,Z) =&\, \sqrt{\frac{\phi\,\eta\,Z^{3/2}}{\kappa}}\tanh\left(\frac{\sqrt{\phi\,\kappa}}{\sqrt{\eta\,Z^{3/2}}}\,t+\arctanh\left(-\frac{\alpha\,\sqrt{\kappa}}{\sqrt{\phi\,\eta\,Z^{3/2}}}\right)\right)\ ,\\
        B(t,Z) =& \,\int_{t}^{T}\beta \exp\left(\int_{s}^{t} \left(\beta-\frac{\kappa}{\eta\,Z^{3/2}}A(u,Z)\right)du\right)ds\,.\\
    \end{split}
\end{equation}
\end{thm}
For a proof see \ref{sec:annex_proof}\,.

{The first term on the right-hand side of the closed-form approximation strategy \eqref{eq:pseudooptimal} is akin to the optimal liquidation rate in the continuous Almgren--Chriss model. This component liquidates the inventory of the LT according to her urgency, her current inventory, the terminal penalty, and the remaining trading time.    The second term is an arbitrage component; it accounts for the spread between the DEX rate $Z$ and the CEX rate $S$. In contrast to the classical literature, the two components also adjust the trading activity of the LT according to the stochastic convexity costs in the pool. In particular, the LT liquidates and speculates more aggressively when the depth of liquidity in the pool is higher; see Section \ref{sec:AMM2_2}.}

\subsection{Comparison with the optimal strategy \label{sec:model:comparison}}
In this subsection, we use an Euler scheme to compute the optimal strategy derived in \ref{apx:Model}, which maximises \eqref{eq:perfcriteria_modelI}, where convexity costs are not piecewise constant. The optimal strategy {is computationally prohibitive because it} uses a {four}-dimensional grid; one dimension is time, one is the rate $\tilde Z,$ one is the oracle rate $\tilde S$, {and one is the inventory $y$}.

Figure \ref{fig:comparisonnumerical} compares the optimal strategy, obtained numerically, and the closed-form approximation strategy \eqref{eq:pseudooptimal} for different values of the inventory $y$, the rate $Z$, and the oracle rate $S$, {when the value of the permanent rate impact parameter is $c=0$. Figure \ref{fig:comparison_perm_impact} compares both strategies for different non-zero values of the permanent impact.} Figure  \ref{fig:comparisonnumerical} indicates that both the closed-form approximation and the optimal strategies are significantly close and both capture the same financial effects. In particular, the strategies clearly depend on the {difference $S-Z$ in rates}  and the inventory $y.$ When the inventory of the LT is zero, the strategy is mostly speculative because the strategy buys asset $Y$ when the oracle rate $S$ is above the rate $Z$, and sells otherwise. When the inventory of the LT is positive (negative), the optimal strategy buys (sells) asset $Y$ only when $S$ is significantly higher (lower) than $Z$. Figure \ref{fig:comparisonnumerical} also shows that the absolute difference between both strategies increases as the difference between the DEX rate $Z$ and the CEX rate $S$ increases, and the difference is minimal when the rates are equal.  {Finally, Figure \ref{fig:comparison_perm_impact} shows that while the absolute difference between the trading speeds increases with the permanent impact, the changes remain negligible even for significantly large values of the permanent impact $c$.}  {These observations hold for other values of the model parameters.}

In practice, by arbitrage, the rates $Z$ and $S$ {in liquid pools} are aligned by LTs trading in the pool, so  differences between both rates are small. For instance,  with the market data we use in Section \ref{sec:Performance}, the average absolute difference between both rates is $2$ USD ($0.07\%$) in the liquid pool and $10$ USD ($0.37\%$) in the illiquid pool we consider.  {However, if the trading activity in the AMM is low and execution costs are high because of fee levels and limited liquidity, then the no-trade intervals (see \cite{he2024liquidity}) may lead to larger discrepancies between the exchange rates than those observed in the pools we consider. Figures  \ref{fig:comparisonnumerical} and \ref{fig:comparison_perm_impact} show that the closed-form approximation strategy departs significantly from the optimal strategy only when the difference between the CEX and DEX rates is exceedingly high, in which case the approximation strategy is not appropriate. However, these differences are seldom observed in AMMs such as Uniswap v3.}

\begin{figure}[!h]\centering
\includegraphics{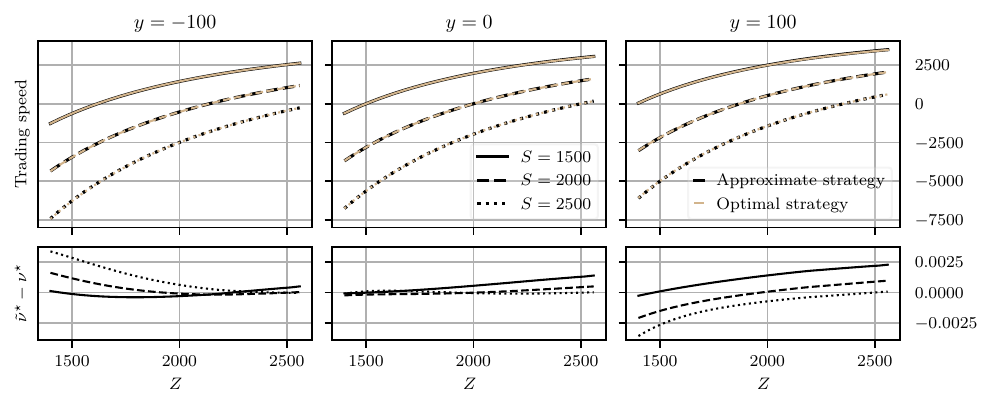}\\
\caption{Comparison of the optimal strategy \eqref{eq:perfcriteria_modelI} obtained with a numerical scheme and the closed-form approximation strategy in \eqref{eq:pseudooptimal}. Model parameters: $T=0.1$, $\sigma = 0.03$, $\gamma = 0.02$, $\beta = 1$, $\alpha = 0.1$, $\kappa = 10^7$, $\eta=1$, $\phi=10^{-5}$, and $Z_0=S_0=2000$. Inventory is $y=-100$ (left panel), $y=0$ (middle panel), and $y=100$ (right panel). {The lower panels show the difference in trading speed in terms of the number of assets $Y$ traded per order if the LT trades at the block creation frequency in Ethereum, i.e., if $\Delta t= 13$ seconds. 
}}
\label{fig:comparisonnumerical}
\end{figure}

{

\begin{figure}[!h]\centering
    \centering
    \subfloat[$c= 0.001$.] {{\includegraphics{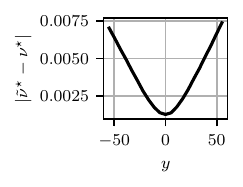} }}%
    \qquad
    \subfloat[$c= 0.01$.] {{\includegraphics{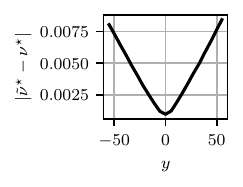} }}%
    \qquad
    \subfloat[$c=0.1$.] {{\includegraphics{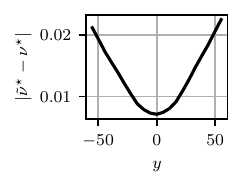} }}%
    \caption{{Mean absolute difference between the optimal strategy and the closed-form approximation strategy for different values of the permanent impact parameter $c$ and for values of $Z$ and $S$ such that $|Z-S|\le 40\% \, Z_0\,.$ For instance, in the ETH/USDC pool, a value $c=0.1$ corresponds to a permanent impact of $0.1$ USDC for each unit of ETH traded, which is too large in practice. The panels show the difference in trading speed in terms of the number of assets $Y$ traded per order if the LT trades at the block creation frequency in Ethereum, i.e., if $\Delta t= 13$ seconds. 
    }}%
    \label{fig:comparison_perm_impact}%
\end{figure}

\section{Optimal trading when rates form in the DEX \label{sec:Model2}}

In the future, activity in AMMs is expected to increase, and so will the informational content of the rates implied by the pool. An increase in activity would affect our modelling choices in two ways. One, the innovations in the depth $\kappa$ will occur more often, so its value cannot be regarded as constant throughout the execution window. Two, the  oracle rate becomes redundant because the rates in the AMM become efficient, so they incorporate all the information available to market participants -- i.e., the discrepancies with rates in other trading venues are negligible and economically insignificant.

In this section, we consider the problem of an LT who wants to exchange a large position in asset $Y$ into asset $X$ in a CPM. The key differences with the model of Section \ref{sec:Model} are that (i) the  marginal rate in the AMM is efficient, so the LT does not use an oracle rate from another venue, and (ii) the AMM implements a CL feature and the activity of LPs is high. Thus, the depth $\kappa$ exhibits frequent and random updates. For simplicity, we omit the permanent impact in the formulation of our problem because, similar to the result of Proposition \ref{prop:perm impact is useless}, one can show that the permanent impact may be ignored to derive the closed-form approximation strategy.  

The LT must liquidate a large position in asset $Y$ over a period of time $[0,T]$. The marginal exchange rate $(Z_t)_{t\ge 0}$, the depth $(\kappa_t)_{t\ge 0}$, the inventory $(y_t)_{t\ge 0}$ in asset $Y$, and the cash process $(x_t)_{t\ge 0}$ evolve as
\begin{equation}
\label{eq:rateProcess_model2}
dZ_{t}=\gamma\,Z_t\,dB_t\, , \quad d\kappa_{t}=\varsigma\, \kappa_{t} \,dL_{t}\,, \quad dy_t = -\nu_t\,dt\,, \quad \text{and} \quad dx_{t} = \left(Z_{t}-\eta\,\frac{Z_{t}^{3/2}}{\kappa_{t}}\,\nu_{t}\right)\,\nu_{t}\,dt\, ,
\end{equation}
where $(B_t)_{t\ge 0}$ and $(L_t)_{t\ge 0}$ are independent standard Brownian motions, and $\varsigma$ is the volatility of the pool depth $\kappa$.

{As in the previous section, the liquidity depth $\kappa$ is not set strategically in our model. In a full equilibrium model involving both LTs and LPs, competition between market participants would endogenously determine the level of liquidity. Here, we abstract from the dynamic strategic behaviour of LPs, treating them as passive. The stochastic variation in $\kappa$ instead reflects the actions of noise LPs who adjust liquidity in a non-strategic manner.}

As in the previous section, the LT's performance criterion is  \eqref{eq:perfcriteria} and similar steps to those of Section \ref{sec:Model}, which we detail in \ref{apx:Model2}, show  that solving the optimal execution problem requires solving  a semilinear PDE which we cannot solve in closed form. Analogous arguments to those of Section \ref{sec:cfapprox} lead to the approximation strategy
\begin{equation}
\label{eq:pseudooptimal Model2}
    \hat{\nu}^{\star}\left(t,{y},\kappa,Z\right)=-y\,\frac{\kappa}{\eta\,Z^{3/2}}\,A(t,\kappa,Z)\,,
\end{equation}
where
\begin{equation}\label{eq:pseudoAB Model2}
    \begin{split}
        A(t,\kappa,Z)=\sqrt{\frac{\phi\,\eta\,Z^{3/2}}{\kappa}}\,\tanh\left(\sqrt{\frac{\phi\,\kappa}{\eta\,Z^{3/2}}}\,(T-t)+\arctanh\left(-\alpha\,\sqrt{\frac{\kappa}{\phi\,\eta\,Z^{3/2}}}\right)\right)\, .
    \end{split}
\end{equation}

The strategy \eqref{eq:pseudooptimal Model2} is a function of the stochastic convexity costs in the AMM. In particular, the strategy liquidates the outstanding inventory at a speed that increases with the depth $\kappa,$ and decreases with the marginal rate $Z.$

\section{Optimal trading when rates form in both the CEX and the DEX \label{sec:Model3}}

When the exchange rates of asset $Y$ in terms of asset $X$ form in both the CEX and the DEX, the dynamics in \eqref{eq:ZProcess} and \eqref{eq:rateProcess_model2} are not appropriate. In particular, we expect both rates to be cointegrated, akin to the joint dynamics observed for equity shares quoted in multiple trading platforms; for more details on cointegration and its applications in algorithmic trading, see \cite{johansen1991estimation, comte1999discrete, bergault2022multi,  drissi2022solvability, cartea2023execution}.

{In this section, we assume that the pair of assets is traded on two AMMs, where we denote the marginal rates by $Z_1$ and $Z_2$, and one CEX, where we denote the quoted rate by $S$. Our model can be generalised to a pair of assets traded on mutiple pools with different fee tiers as in Uniswap v3. The joint dynamics $\boldsymbol{P}=\left(Z_1,Z_2,S\right)^{\intercal}$ of the marginal rates follow the three-dimensional Ornstein--Uhlenbeck dynamics 
\begin{equation}\label{eq:model3 dynamics P}
d\boldsymbol{P}_{t}=\boldsymbol{\Pi}\,\left(\boldsymbol{\overline{P}}-\boldsymbol{P}_{t}\right)\,dt+\boldsymbol{V}\,d\boldsymbol{W}_{t}\,,
\end{equation}
where $\boldsymbol{\Pi}$ is a $3\times3$ mean-reversion matrix, $\boldsymbol{\overline{P}} \in \R^3$ is the long-term unconditional mean value of both rates, $\boldsymbol{V}$ is the Cholesky decomposition of the correlation matrix $\boldsymbol{\Sigma}$ of both rates, and $\boldsymbol{W}$ is a three-dimensional standard Brownian motion. In the dynamics \eqref{eq:model3 dynamics P}, the matrix  $\boldsymbol{\Pi}$ drives the deterministic part of the joint dynamics of the CEX and DEX rates, i.e., the speed at which both rates are aligned and at which they converge to their long-term mean $\boldsymbol{\overline P}.$ 

The inventory of the LT follows the dynamics in \eqref{eq:ytildeProcess_modelI}, her cash follows the dynamics in  \eqref{eq:cash}, and her performance criterion is in \eqref{eq:perfcriteria_modelI}. 
Here, we derive an approximation strategy that can be implemented in real time by the LT. More precisely, we assume constant convexity costs \eqref{eq:impact_deterministic} and the alternative performance criterion \eqref{eq:perfcriteria}. Following similar arguments as those in Section \ref{sec:Model} which we detail in \ref{apx:Model3}, we obtain the closed-form approximation strategy
\begin{equation}\label{eq:model3:approx_strategy}
\hat\nu^{\star}(t,y,\boldsymbol{P})=-y\,\frac{\kappa}{\eta\,\boldsymbol{\mathcal{X}}\,\boldsymbol{P}^{3/2}}\,A(t,\boldsymbol{P})+\frac{1}{2\,\eta\,\boldsymbol{\mathcal{X}}\,\boldsymbol{P}^{3/2}}\,B(t,\boldsymbol{P})\,\left(\boldsymbol{\overline{P}}-\boldsymbol{P}\right)\,,
\end{equation}
when rates form in both the CEX and the DEX, where 
\begin{equation}
\begin{cases}
A(t,\boldsymbol{P})= & \sqrt{\frac{\phi\,\eta\,\boldsymbol{\mathcal{X}}\,\boldsymbol{P}^{3/2}}{\kappa}}\,\tanh\left(\sqrt{\frac{\phi\,\kappa}{\eta\,\boldsymbol{\mathcal{X}}\,\boldsymbol{P}^{3/2}}}\,(T-t)+\arctanh\left(-\alpha\,\sqrt{\frac{\kappa}{\phi\,\eta\,\boldsymbol{\mathcal{X}}\,\boldsymbol{P}^{3/2}}}\right)\right)\,,\\
B(t,\boldsymbol{P})= & -\int_{t}^{T}e^{\int_{t}^{s}\frac{\kappa}{\eta\,\boldsymbol{\mathcal{X}}\,\boldsymbol{P}^{3/2}}\,A(u,\boldsymbol{P})\,du}\,\boldsymbol{\mathcal{X}}\,\boldsymbol{\Pi}^{\intercal}\,e^{-\boldsymbol{{\Pi}}^{\intercal}(s-t)}\,ds\,,
\end{cases}
\end{equation}
where $\boldsymbol{\mathcal{X}}$ is an indicator matrix which identifies the trading venues in which the agent executes trades. 

Similar to  strategy \eqref{eq:pseudooptimal}, when rates form in the CEX, strategy \eqref{eq:model3:approx_strategy} relies on two components. The first term on the right-hand side of \eqref{eq:model3:approx_strategy} is linear in the inventory of the LT. When the convexity costs decrease, i.e., when the depth of liquidity in the pool increases, this component liquidates the inventory more aggressively. The second term on the right-hand side of \eqref{eq:model3:approx_strategy} is a speculative component that trades according to the difference between the joint CEX/DEX rates or the joint DEX/DEX rates $\boldsymbol{P}$, and their long-term unconditional mean $\overline{\boldsymbol{P}}$.}

}

\section{Performance of strategies \label{sec:Performance}}

{Here, we use Uniswap v3 data for the liquid pool ETH/USDC and the illiquid pool ETH/DAI. At present, rates form mainly in the CEX Binance as detailed in \ref{sec:apx:dataanalysis}. Thus, }we study the performance of the closed-form approximation strategy in \eqref{eq:pseudooptimal} and \eqref{eq:pseudoAB},  {which corresponds to when rate formation is exclusive to CEXs}. We consider two setups. One focuses on liquidating a large position in one asset and the other uses the lead-follow relationship between the oracle and AMM rates to execute a statistical arbitrage. {In this section,} {we account for AMM and gas fees and assume that the orders sent by the LT do not impact the dynamics of the pools.}

We use in-sample data to estimate model parameters and use out-of-sample data to execute the strategies. For in-sample data, we use a {rolling} window of 24 hours prior to the start of the trading programme for both pools. For out-of-sample data, we use {rolling} windows of $2$ and $12$ hours when the LT trades in the liquid and illiquid pools, respectively. To measure performance, we use rolling time windows between 1 July 2021 and 31 December 2023 for estimation and execution.  Specifically, after every execution programme, we shift both windows by $2$ and $12$ hours for the liquid and illiquid pools, respectively, and repeat the same procedure, i.e., estimate parameters with in-sample data and trade with out-of-sample data. We remark that we do not simulate rates, we use those of the AMM and Binance, and the execution costs are those the trades would have received. In total, we run { 8,579 and 1,747} execution programmes for ETH/USDC and ETH/DAI, respectively.

We proceed as follows. Subsection \ref{subsec: liquidation stategy} describes how parameter estimates are obtained and showcases the performance of the liquidation strategy.  Subsection \ref{subsec: speculation stategy} showcases the performance of the statistical arbitrage strategy.

\subsection{Liquidation strategy}\label{subsec: liquidation stategy}

We describe how to estimate the in-sample model parameters for every run of the strategy. For rate dynamics, the LT performs OLS regressions on the discretised versions  
\begin{align}\label{eq: discrete rate dyn}
\Delta \log S_t =& -\frac{\sigma^2}{2} \, \Delta t + \sigma  \,\sqrt{\Delta t} \, \upsilon_t\,,
\nonumber \\
\Delta \log Z_t  =& -\frac{\gamma^2}{2} \, \Delta t + \beta\, \left(\frac{S_t - Z_t}{Z_t}\right) \,\Delta t + \gamma \, \sqrt{\Delta t} \, \epsilon_t\,,
\end{align}
of \eqref{eq:SProcess} and \eqref{eq:ZProcess}, where  $\{\epsilon_t\,,\upsilon_t\}$ are error terms. Here, the size of the time-step $\Delta t$ is the frequency of the liquidity taking orders (from all LTs) that arrive in the pool during the estimation period.

For the liquidation strategy, we target a participation rate of $50\%$ of the observed hourly volume to set the initial inventory, which is liquidated by the LT over the trading window at the same frequency as the observed average trading frequency over the in-sample estimation period. The trading frequency determines the value of the parameter $\eta$ in \eqref{eq:execratesApproxCPM_nu}, which scales the convexity costs.

The value of the other model parameters are as follows. The value of parameter $\kappa$ is the last observed depth of the pool before the start of the execution. The value of the running inventory parameter $\phi$  is kept constant for all runs. The value of the terminal penalty parameter $\alpha$ is arbitrarily large to enforce full liquidation of outstanding inventory by the end of the trading horizon. {For all strategy runs, $T=2$ hours and $\alpha=10\,  \textrm{USDC}\cdot\textrm{ETH}^{-2}$ for ETH/USDC, and $T=12$ hours and $\alpha=10\,  \textrm{DAI}\cdot\textrm{ETH}^{-2}$ for ETH/DAI.} Finally, as a more detailed example, \ref{Example one run liquidation} describes parameter estimation and performance for a specific run of the liquidation strategy.

We benchmark the performance of the liquidation strategy with two strategies: TWAP, which consists in trading at a constant rate; and a single order execution strategy, which consists in executing the entire order at the beginning of the execution window.  The market rates at the time of trading are used to compute the convexity costs for all strategies. Gas fees are $5$ USD per transaction, regardless of transaction size. On the other hand, AMM fees depend on transaction size, {and here we impute a $0.05 \, \%$ fee for ETH/USDC and a $0.3\%$ fee for ETH/DAI to the value of every transaction}.  Tables \ref{table:perfsUSDC} and \ref{table:perfsDAI} show the average and standard deviation of the gross PnLs which is given by $x_T + y_T \, Z_T - y_0 \, Z_0,$ the number of transactions, the gas fees, and the AMM fees.\footnote{Gross PnL, as opposed to net PnL, is computed without the AMM fees and gas fees paid by the LT.}

{\footnotesize
\begin{table}[h]
\scriptsize
\parbox{.45\linewidth}{
\centering
\begin{tabular}{c  r  r   r  r} 
\hline 
 &  Gross avg.  & Std. & Avg. & Avg.   \\
 & PnL  & dev. & $\#$trades &  fees  \\ [0.5ex]
\hline
Single order & $-$827,692 & 2,614,863 & 1 & 10,575\\ [1ex] 
TWAP         &  919 &  206,538 & 483 & 10,575\\ [1ex] 
Liquidation  &   \\ 
$(\phi=0.01)$  &  11,337 &  232,492 & 483 & 16,832 \\ 
Liquidation  &   \\ 
$(\phi=0.005)$  &  15,988 &  291,716 & 483 & 17,996 \\  
Liquidation  &   \\ 
$(\phi=0.001)$  &  23,525 &  404,636  & 483 & 21,165 \\ 
\hline 
Speculative  &   \\ [-0.4ex]
$(y_0=0)$  &     23,818 &  389,603 & 483 & 15,096 \\ [-0.45ex]
$(\phi=0.001)$ & \\ 
\hline 
\end{tabular}
\caption{Performance and fees for ETH/USDC. The Average PnL does not include fees. {The performance is based on 8,579 runs.}}
\label{table:perfsUSDC}
}
\hfill
\parbox{.45\linewidth}{
\centering
\begin{tabular}{c  r  r  r  r} 
\hline 
 &  Gross avg.  & Std. & Avg.  & Avg.   \\
 & PnL  & dev. & $\#$trades &  fees  \\ [0.5ex]
\hline
Single order & $-$148,352 & 613,652 & 1 & 8,527\\ [1ex] 
TWAP         &   $-$2,233 &  92,468 & 84 & 9,917 \\ [1ex] 
Liquidation  &   \\ 
$(\phi=0.01)$  &  849 &  48,627 & 84 & 9,895 \\ 
Liquidation  &   \\ 
$(\phi=0.005)$  &  1,449 &  58,177 & 84 & 9,923\\  
Liquidation  &   \\ 
$(\phi=0.001)$  &  2,164 &  79,448 & 84 & 10,403 \\ 
\hline 
Speculative  &   \\ [-0.4ex]
$(y_0=0)$  &    2,671 &  14,180 & 84 & 2,381 \\ [-0.45ex]
$(\phi=0.001)$ & \\ 
\hline 
\end{tabular}
\caption{Performance and fees for ETH/DAI. The Average PnL does not include fees. {The performance is based on 1,747 runs.}}
\label{table:perfsDAI}
}
\end{table}
}

Tables \ref{table:perfsUSDC} and \ref{table:perfsDAI} show that liquidating all the inventory in one trade is sub-optimal compared with the other strategies due to the high execution costs of the large order. In both cases, our model outperforms TWAP in terms of the ratio between performance, net of fees, and risk measured by the standard deviation. Key to the outperformance is that the liquidation strategy uses the rates in Binance as a trading signal.  {Finally, as the value of the penalty $\phi$ for holding inventory decreases, the speculative component becomes prominent, so the LT executes more CEX/DEX arbitrages, resulting in an increased PnL but also in increased risk and execution fees.}

\subsection{Speculative strategy \label{subsec: speculation stategy}}
We consider the same setup as before, i.e., the in-sample estimation and out-of-sample execution. Here, the LT arbitrages the AMM. To this end, the LT starts with zero inventory in $Y$ and sets the values of the urgency parameter $\phi$ to $0.001 \ \textrm{USDC}\cdot\textrm{ETH}^{-2}$ and $0.001 \ \textrm{DAI}\cdot\textrm{ETH}^{-2}$ for the liquid and illiquid pools, respectively. The strategy profits from the oracle rate as a predictive signal. The last rows in Tables \ref{table:perfsUSDC} and \ref{table:perfsDAI} show the average and standard deviation of the gross PnLs, the number of transactions, and the estimated AMM and gas fees.

\section{Conclusions \label{sec:Conclusion}}

In this work, we used Uniswap v3 data to analyse rate, liquidity, and execution costs of CPMs. We proposed a model for {execution costs and for} optimal trading in CPMs where we assumed that rates are formed in an alternative CEX and the liquidity provided in the CPM is constant for relevant periods of time. {Also, we proposed models when rates are efficient in the DEX, or when rates form in both the DEX and in the CEX.}  Finally, we used in-sample estimation of model parameters and out-of-sample market data to test the performance of a closed-form approximation of the optimal strategy, so our results do not rely on simulations. We showed that our strategy considerably outperforms TWAP and a strategy that consists in sending a single large order. We also showed that there are significant arbitrage opportunities between Binance and AMM rates.

Our models consider CPMs with CL and assume zero impact on the strategic behaviour of LPs. Future work should explore strategic liquidity provision; see  \cite{fan2021strategic}, \cite{neuder2021strategic}, \cite{fan2022differential}, \cite{fukasawa2023model}, \cite{li2023yield},   \cite{lommers2023:case}, \cite{goyal2023finding}, and \cite{capponi2023paradox} who study different aspects of liquidity provision in AMMs.  Finally, there is a growing literature on the design of AMMs. For instance, \cite{evans2021optimal} study optimal fees in geometric markets, \cite{cartea2023automated,cartea2024strategic} generalise CFMs and propose dynamic fees where LPs express their risk preferences, and \cite{goyal2023finding} study AMMs with dynamic trading functions; see also \cite{bergault2022automated, sabate2022variable, cohen2023inefficiency, curry2024optimal, wood2024expiring, alexander2024theoretical}. The applicability of our models in these AMMs requires a careful analysis of execution costs.

\clearpage
\appendix
\section{Uniswap v3: pool depth and oracle rates \label{sec:apx:dataanalysis}}

Transaction information from decentralised exchanges is public. 
Here, we analyse transaction data of the most liquid pools of Uniswap v3 for the two pairs of assets ETH/USDC and ETH/DAI. {The ETH/USDC pool is the Uniswap v3 pool with the address \texttt{0x88e6a0c2ddd26feeb64f03
9a2c41296fcb3f5640} which charges a $0.05\%$ proportional fee, and the ETH/DAI pool is the Uniswap v3 pool with the address \texttt{0xc2e9f25be6257c210d7adf0d4cd6e3e881ba25f8} which charges a $0.3\%$ proportional fee.} These pools are considered an alternative to LOB-based trading venues such as Binance, which is the most liquid and active venue for both pairs. 
The ticker ETH represents the cryptocurrency \textit{Ether}, which is the native cryptocurrency of the Ethereum blockchain. The ticker USDC represents \textit{USD coin}, a cryptocurrency fully backed by U.S. Dollars (USD); and DAI represents the cryptocurrency \textit{Dai}, which tracks parity with the U.S. Dollar.

\begin{table}[!h]
    \centering
    \begin{footnotesize}
    \begin{tabular}{l|rrr|rrr}
    \hline
    & \multicolumn{3} {c|} {ETH/USDC $0.05\%$} & \multicolumn{3} {c} {ETH/DAI $0.3\%$} \\
    & \multicolumn{3} {c|} {(6.76$\times\text{10}^\text{6}$ LT transactions} & \multicolumn{3} {c}  {(218,045 LT transactions)} \\
    & \multicolumn{3} {c|} {200,490 LP transactions} & \multicolumn{3} {c}  {21,261 LP transactions} \\
    & \multicolumn{3} {c|} {17.53$\times\text{10}^\text{6}$  Binance transactions)} & \multicolumn{3} {c}  {5.92 $\times\text{10}^\text{6}$ Binance transactions)} \\
    \hline
     &          mean &    median &  std. dev.
    &          mean &    median &  std. dev.    \\
    \hline 
    LT transaction size &&&&&& \\ [-0.5ex]
     (USD)   &  69,615 &  4,95 & 235,780 &  65,035 &  29,259 & 125,550 \\ [0.5ex]

    LT trading frequency &&&&&&\\ [-0.5ex]
     (seconds) &  14.32 & 12 &  23.08 &  440 &  90 & 1282\\  [0.5ex]

    LT gas fee &&&&&& \\ [-0.5ex]
     (USD) &  67.52 & 19.39 &  346.95 &  131.19 &  36.79 & 611.54 \\ [0.5ex]

    LT unitary execution costs \eqref{eq:executionCost} &&&&&&\\ [-0.5ex]
     (USD)  &  0.25 & 0.014 &  1.137 &  1.54 &  0.87 & 4.34 \\ [0.5ex]
    
    LP transaction size &&&&&& \\ [-0.5ex]
     (USD) & 7,560,310 & 137,890 & 12,457,171 & 575,246 & 1,355 & 3,255,062 \\ [0.5ex]

    LP trading frequency &&&&&&\\ [-0.5ex]
     (seconds) &  486 & 126 &  14,713 &  4,588 &  1,082 & 9,496\\  [0.5ex]
    
    LP gas fee &&&&&& \\ [-0.5ex]
     (USD) &  45.53 & 19.90 &  124.6 &  58.55 &  20.21 & 246.29 \\ [0.5ex]
    
    pool size &&&&&&\\ [-0.5ex] 
     (USD $\text{10}^\text{6}$)   &  2,655 &  1,618 & 456,090  &  294 &   	139 & 580 \\ [0.5ex]
    \hline
    Binance transaction size &&&&&&\\ [-0.5ex] 
     (USD)   &  1184 &  382 & 1668  &  921 &  411 & 718 \\ [0.5ex]
    Binance trading frequency &&&&&&\\ [-0.5ex] 
     (USD)   &  3.69&  0.24 & 11  &  13 &  0.77 & 59 \\ [0.5ex]
    \hline
    \end{tabular}
    \end{footnotesize}
    \caption{{Descriptive statistics for LT and LP trading activity in the ETH/USDC and ETH/DAI pools between 5 May 2021 and 31 December 2023}. }
    \label{tab:descriptive_statistics}
\end{table}

{Table \ref{tab:descriptive_statistics} shows that LT trading activity, LP trading activity, and the depth of liquidity in the pool ETH/USDC $0.05\%$ are significantly larger than those in the pool ETH/DAI $0.3\%.$ As a consequence, execution costs are lower and LT and LP transaction sizes are smaller in the ETH/USDC pool. Table \ref{tab:descriptive_statistics} also shows that the trading activity in Binance is considerably higher that that in Uniswap v3. In particular, trading is at a significantly higher frequency and transaction sizes are smaller.

Next, we study the suitability of our assumption of constant pool depth $\kappa$ in the optimal trading model of  Section \ref{sec:Model}. Recall that, compared with most CPMs, Uniswap v3 operates with the CL feature. LPs in Uniswap v3 specify the range of rates where they supply liquidity. Therefore,  with CL, the pool is characterised by the distribution of liquidity across ranges of rates.   First, Table \ref{tab:descriptive_statistics}  shows that LP trading activity is significantly lower than LT trading activity in both pools; LTs trade at a higher frequency and the cumulative size of their trading is higher.\footnote{For instance, LPs executed a transaction every $8$ minutes in the ETH/USDC pool, the most active pool in Uniswap v3, and every $76$ minutes in the ETH/DAI pool.} Thus, at present, the distribution of liquidity in  CL pools may be assumed constant for the trading horizons that we consider. Second, due to the CL feature of Uniswap v3, the depth $\kappa$ of the pool may change when the marginal rate crosses the boundary of a tick. In particular, when the volume of an LT transaction is large enough to make the marginal rate cross a tick where the level of liquidity changes, the AMM treats it as multiple trades, each with a different value of $\kappa$. In our data, most of the liquidity is concentrated around the marginal rate and the depth is the same over a large range around the rate; see \cite{drissi2023models} for more details.  
Thus, one  may assume a constant pool depth $\kappa$ because of the low LP trading activity and the shape of the distribution of liquidity around the marginal rate in Uniswap v3.

}

{
\section{Proof of Proposition \ref{prop:perm impact is useless} \label{apx:proofprop1}}
First, use integration by parts to write\begin{align*}
&\tilde{x}_{T}+y_{T}\tilde{Z}_{T}-\alpha\,{y_{T}}^{2}\\=\,&x_{0}+\int_{0}^{T}\nu_{t}\,\tilde{Z}_{t}\,dt-\int_{0}^{T}\eta\,\zeta\,\nu_{t}^{2}\,dt+y_{T}\,\tilde{Z}_{T}-\alpha\,y_{T}^{2}
\\=&x_{0}+\int_{0}^{T}y_{t}\,d\tilde{Z}_{t}+y_{0}\,Z_{0}-y_{T}\,\tilde{Z}_{T}-\int_{0}^{T}\eta\,\zeta\,\nu_{t}^{2}\,dt+y_{T}\,\tilde{Z}_{T}-\alpha\,y_{T}^{2}
\\=&\,x_{0}+y_{0}\,Z_{0}+\int_{0}^{T}y_{t}\,dZ_{t}+\frac{c}{2}\,\left(y_{T}^{2}-y_{0}^{2}\right)-\int_{0}^{T}\eta\,\zeta\,\nu_{t}^{2}\,dt-\alpha\,y_{T}^{2}\,.
\end{align*}
Next, use the dynamics \eqref{eq:unaffected cash dynamics} to write
$$x_{T}+y_{T}\,Z_{T}=x_{0}+\int_{0}^{T}\nu_{t}\,Z_{t}dt-\int_{0}^{T}\eta\,\zeta\,\nu_{t}^{2}dt+y_{T}\,Z_{T}=x_{0}+y_{0}\,Z_{0}+\int_{0}^{T}y_{t}dZ_{t}-\int_{0}^{T}\eta\,\zeta\,\nu_{t}^{2}\,dt\,,$$
so
$$\tilde x_{T}+y_{T}\,\tilde{Z}_{T}-\alpha\,y_{T}^{2}=x_{T}+y_{T}\,Z_{T}-\frac{c}{2}\,y_{0}^{2}-\left(\alpha-\frac{c}{2}\right)\,{\tilde{y}_{T}}^{2}\,,$$
which proves the result.
\qed
}

\section{Proof of Theorem \ref{thm}}\label{sec:annex_proof}

Recall that for each fixed values of  $N$ and $j$
\begin{equation}
\nu^{\star,j,N}\left(t,y,Z,S\right)=-\frac{1}{\eta\,\zeta_{N}^{j}}A_{j,N}(t)\,y+\frac{1}{2\,\eta\,\zeta_{N}^{j}}B_{j,N}(t)(S-Z)\ ,
\end{equation}
where
\begin{equation}
    \zeta_{N}^{j}\coloneqq \frac{1}{\kappa}\left(Z^N_j\right)^{3/2}\,,
\end{equation}
and
\begin{equation}
    \begin{split}
        &A_{j,N}(t)\coloneqq A_{\zeta_j^N}(t)=\sqrt{\phi\,\eta\,\zeta_{N}^{j}}\tanh\left(\frac{\sqrt{\phi}}{\sqrt{\eta\,\zeta_{N}^{j}}}\,t+\arctanh\left(-\frac{\alpha}{\sqrt{\phi\,\eta\,\zeta_{N}^{j}}}\right)\right)\ ,\\
        &B_{j,N}(t)\coloneqq B_{\zeta_j^N}(t) = -\int_{t}^{T}\beta \exp\left(-\int_{t}^{s} \left(\beta-\frac{1}{\eta\,\zeta_{N}^{j}}\,A_{j,N}(u)\right)du\right)ds\,.
    \end{split}
\end{equation}

Moreover, recall that 
\begin{equation}
    \begin{split}
        A(t,Z) =& \sqrt{\frac{\phi\,\eta\,Z^{3/2}}{\kappa}}\tanh\left(\frac{\sqrt{\phi\,\kappa}}{\sqrt{\eta\,Z^{3/2}}}\,t+\arctanh\left(-\frac{\alpha\,\sqrt{\kappa}}{\sqrt{\phi\,\eta\,Z^{3/2}}}\right)\right)\ ,\\
        B(t,Z) =& \int_{t}^{T}\beta \exp\left(-\int_{t}^{s} \left(\beta-\frac{\kappa}{\eta\,Z^{3/2}}\,A(u,Z)\right)du\right)ds\,.\\
    \end{split}
\end{equation}

To prove \eqref{eq:inequality}\,, take $(t,y,S)$ and write
\begin{equation}
\begin{split}
    \big|\nu^{\star,j,N}\big(t,y&,Z^N_{j+1},S\big)-\nu^{\star,j+1,N}\big(t,y,Z^N_{j+1},S\big)\big| \\ =&\left|-\frac{1}{\eta\,\zeta_{N}^{j}}A_{j,N}(t)\,y
    + \frac{1}{\eta\,\zeta_{N}^{j+1}}A_{j+1,N}(t)\,y + (S-Z^N_{j+1})\left(\frac{1}{2\,\eta\,\zeta_{N}^{j}}B_{j,N}(t)-\frac{1}{2\,\eta\,\zeta_{N}^{j+1}}B_{j+1,N}(t)\right)\right|\\
    \leq&\,\frac{|y|}{\eta}\left|-\frac{1}{\zeta_{N}^{j}}A_{j,N}(t)
    + \frac{1}{\zeta_{N}^{j+1}}A_{j+1,N}(t) \right|+\frac{\left|S\right|+\overline{Z}}{\eta}\, \left| \frac{1}{2\,\zeta_{N}^{j}}B_{j,N}(t)-\frac{1}{2\,\zeta_{N}^{j+1}}B_{j+1,N}(t)\right|\\
    =&\,\frac{|y|}{\eta}\left|-\frac{\kappa}{\left(Z^N_j\right)^{3/2}}A\left(t,Z^N_j\right)
    + \frac{\kappa}{\left(Z^N_{j+1}\right)^{3/2}}A\left(t,Z^N_{j+1}\right) \right|\\
    &+\frac{\left|S\right|+\overline{Z}}{2\,\eta}\, \left| \frac{\kappa}{\left(Z^N_j\right)^{3/2}}B\left(t,Z^N_j\right)-\frac{\kappa}{\left(Z^N_{j+1}\right)^{3/2}}B\left(t,Z^N_{j+1}\right)\right|\,.
\end{split}
\end{equation}

Observe that for a fixed $t\in[0,T]$ the functions
\begin{equation}
    \begin{split}
        Z\mapsto\frac{\kappa}{Z^{3/2}}A\left(t,Z\right)\, \ \ \text{and} \ \ 
        Z\mapsto\frac{\kappa}{Z^{3/2}}B\left(t,Z\right)
    \end{split}
\end{equation}
are uniformly continuous on $\left[\underline{Z},\overline{Z}\right]$ because they are both compositions of continuous functions defined over a closed interval. By definition of the partition in \eqref{eq:partitionedConvexity}, $\left|Z^N_{j} - Z^N_{j+1}\right| = 1/N $ so for each $\varepsilon>0$ there exists $N\in\N$ such that
\begin{equation}
    \max_{j=1,\dots,N}\big|\nu^{\star,j,N}\big(t,y,Z^N_{j+1},S\big)-\nu^{\star,j+1,N}\big(t,y,Z^N_{j+1},S\big)\big|\leq\varepsilon\,.
\end{equation}

To prove that $\{\hat{\nu}^{\star,N}\}$ converges uniformly to $\hat{\nu}^{\star}$ in $ [0,T]\times\R\times\left[\underline{Z},\overline{Z}\right]\times\R $\,, take $(t,y,Z,S)\in[0,T]\times\R^{2}\times\left[\underline{Z},\overline{Z}\right]\times\R\,, $ and $N\in\N\,,$ and observe that there exists $\overline{j}\in\{1,\dots,N\}$ such that $Z\in\left[Z_{\overline{j}}^N,Z^N_{\overline{j}+1}\right)$ and thus
\begin{equation}
\begin{split}
    \big|\hat{\nu}^{*}_{N}\big(t,y&,Z,S\big)-\hat{\nu}^{*}\big(t,y,Z,S\big)\big| \\
    =&\big|\hat{\nu}^{*}_{\overline{j},N}\big(t,y,Z,S\big)-\hat{\nu}^{*}\big(t,y,Z,S\big)\big|\\
    =&\left|-\frac{1}{\eta\,\zeta_{N}^{\overline{j}}}A_{\overline{j},N}(t)\,y
    + \frac{\kappa}{\eta\,Z^{3/2}}A(t,Z)\,y + (S-Z)\left(\frac{1}{2\,\eta\,\zeta_{N}^{\overline{j}}}B_{\overline{j},N}(t)-\frac{\kappa}{2\,\eta\,Z^{3/2}}B(t,Z)\right)\right|\\
    \leq&\,\frac{|y|\,\kappa} {\eta\,\underline{Z}^{3/2}}\left|-A\left(t,Z^N_{\overline{j}}\right)
    + A\left(t,Z\right) \right|+ \frac{\left|S\right|+\overline{Z}} {2\,\eta\,\underline{Z}^{3/2}}\,\kappa\, \left| B\left(t,Z^N_{\overline{j}}\right)-B\left(t,Z\right)\right|\,.
\end{split}
\end{equation}

The uniform convergence of $\{\hat{\nu}^{*}_{N}\}$ to $\hat{\nu}^*$ follows from the uniform continuity of $A(t,Z)$ and $B(t,Z)$ on $[0,T]\times\left[\underline{Z},\overline{Z}\right]\,.$

\section{Example for the liquidation strategy \label{Example one run liquidation}}
In this appendix, we describe the parameters and strategy performance for a specific run of the liquidation strategy. Assume the LT will start trading at noon on 16 March 2022, so she uses the data between noon 15 March 2022 and noon 16 March to estimate model parameters. 

For the 24 hours before noon 16 March 202, there are, on average,  one liquidity taking order every 13 seconds in the liquid pool and one every 360 seconds in the illiquid pool; i.e.,  the time steps in the regressions \eqref{eq: discrete rate dyn} are $\Delta t = 13$ for ETH/USDC and $\Delta t = 360$ for ETH/DAI. Table \ref{table:params1} shows parameter estimates. 

{\footnotesize
\begin{table}[!h] 
\begin{center}
\begin{tabular}{c  c  c} 
 \hline 
 & ETH/USDC & ETH/DAI \\ [0.5ex] 
 \hline
 $\hat \sigma$ & $0.045\ \textrm{day}^{- 1/2}$ & $0.053\ \textrm{day}^{- 1/2}$ \\ [0.5ex] 
$\hat\gamma$ & $0.034\ \textrm{day}^{- 1/2}$& $0.027\ \textrm{day}^{- 1/2}$ \\ [0.5ex]
$\hat\beta$ & $657.9\ \textrm{day}^{-1}$ & $14.78\ \textrm{day}^{-1}$  \\  [0.5ex]
 \hline 
\end{tabular}
\end{center}
\caption {Parameter estimates for dynamics of $Z$ and $S$ with data between noon 15 March 2022 and noon 16 March 2022.}
\label{table:params1}
\end{table}
}

The parameter $\eta$ of the execution costs in \eqref{eq:execratesApproxCPM_nu} is also set to $13\, \textrm{seconds} \, = \, 17.3\times 10^{-5} \,\textrm{days}$ and $360\, \textrm{seconds} \,=\, 41 \times 10^{-4}\, \textrm{days}$ for the liquid and illiquid pool, respectively. The number of  transactions in the in-sample data is approximately 238,039 ETH and 4,031 ETH in the liquid and illiquid pool, respectively. Thus the LTs' target is to liquidate 14,877 and 1,007 units of ETH within $2$ and $12$ hours in the ETH/USDC and ETH/DAI pools, respectively. Table \ref{table:params2} summarises all the parameters used to run our strategy.

{\footnotesize
\begin{table}[!h]
\begin{center}
\begin{tabular}{c  r  r} 
\hline 
& ETH/USDC & ETH/DAI \\ [0.5ex] 
\hline
$\kappa_0$ & 22,561,783 & 1,666,175 \\[0.5ex] 
$y_0$ & 14,877 ETH & 1,007 ETH\\ [0.5ex]
$S_0$ & 2,689.2$\  \textrm{USDC}$ & 2,686.09$\  \textrm{DAI}$ \\ [0.5ex]
$Z_0$ & 2,690.77$\  \textrm{USDC}$ & 2,694.04$\  \textrm{DAI}$ \\ [0.5ex]
$\eta$ & $17.3\times 10^{-5}$ $\ \textrm{days}$ & $41 \times 10^{-4}$ $\ \textrm{days}$ \\
\hline 
\end{tabular}
\end{center}
\caption {Values of model parameters.}
\label{table:params2}
\end{table}
}

Figure \ref{fig:backtest_USDC_DAI} shows the marginal and oracle rates and the inventories of the strategies during the execution, for both ETH/USDC and ETH/DAI. Figure \ref{fig:backtest_USDC_DAI} clearly showcases the difference between the strategies. In particular, the liquidation strategy is speculative and trades on the difference between the two rates $S$ and $Z$ during the liquidation programme. Figure \ref{fig:backtest_USDC_DAI_speed} shows how the difference $S_t - Z_t$ drives the trading speed $\nu_t.$ The oracle rate is used as a predictive signal for future moves of the marginal rate.
\begin{figure}[!h]
    \centering
    \subfloat[\textbf{Top}: Out-of-sample marginal and oracle rates for the pair ETH/USDC. \textbf{Bottom}: Inventory process $y$ for the optimal, TWAP, and single order strategies.]  {{\includegraphics{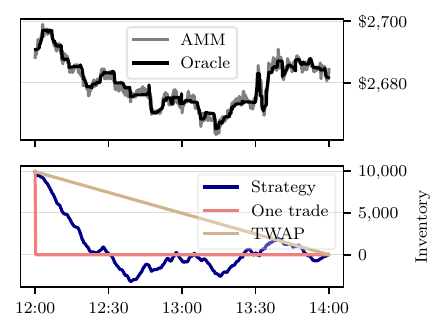} }}%
    \qquad
    \subfloat[\textbf{Top}: Out-of-sample marginal and oracle rates for the pair ETH/DAI. \textbf{Bottom}: Inventory process $y$ for the optimal, TWAP, and single order strategies.] {{\includegraphics{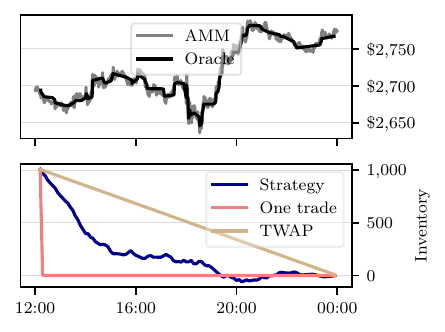} }}%
    \caption{Liquidation strategies starting at noon on 16 March 2022.}%
    \label{fig:backtest_USDC_DAI}%
\end{figure}

\begin{figure}[!h]
    \centering
    \subfloat[\textbf{Top}: Difference between the oracle and the marginal rate for ETH/USDC. \textbf{Bottom}: $-\nu_t.$] {{\includegraphics{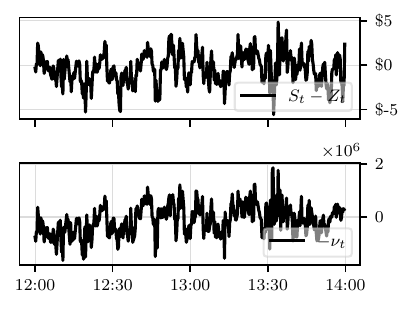} }}%
    \qquad
    \subfloat[\textbf{Top}: Difference between the oracle and the marginal rate for ETH/DAI. \textbf{Bottom}: $-\nu_t.$] {{\includegraphics{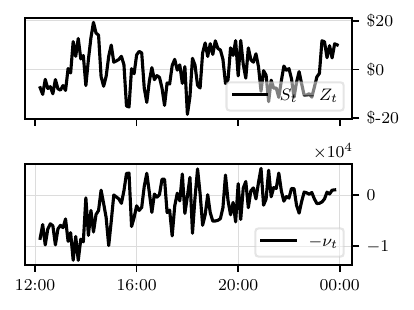} }}%
    \caption{Trading speed.}%
    \label{fig:backtest_USDC_DAI_speed}%
\end{figure}

{
\section{Model of Section \ref{sec:Model} \label{apx:Model}}

For the three models in Sections \ref{sec:Model}, \ref{sec:Model2}, \ref{sec:Model3}, we fix a filtered probability space  $\left(\Omega, \mathcal F, \P; \F = (\mathcal F_t)_{t\ge 0} \right)$ satisfying the usual conditions, where $\F$ is the natural filtration generated by the collection of observable (progressively measurable) stochastic processes that we define for each model below.

\subsection{A {semilinear PDE} \label{apx:originalmodel}}

The set of admissible strategies is
\begin{align}
\label{def:admissibleset_t_modelI}
\mathcal A_t = \left\{ (\nu_s)_{s \in [t,T]},\ \R\textrm{-valued},\ \F\textrm{-adapted, and } \int_t^T |\nu_s|^2 \,ds < +\infty, \ \ \P\textrm{-a.s.}  \right\}\,.
\end{align}
Write $\mathcal A := \mathcal A_0$ and let $\nu\in\Ac$. The LT's value function is
\begin{align}
\label{eq:valuefunc_modelI}
    u(t,\tilde x,y,\tilde Z, \tilde S)=\underset{\nu\in\mathcal{A}}{\sup} \, u^{\nu}\left(t,\tilde x,y,\tilde Z,\tilde S\right)\,,
\end{align}
and it solves the Hamilton--Jacobi--Bellman (HJB) equation{\begin{equation}\label{eq:ANX:hjbu}
    \begin{split}
        0=\,&\partial_{t}w-\phi\,y^{2}+\beta\,\left(\tilde S -\tilde Z\right)\,\partial_{\tilde Z}w+\frac{1}{2}\,\gamma^{2}\,\tilde Z^{2}\,\partial_{\tilde Z\tilde Z}w+\frac{1}{2}\,\sigma^{2}\,\tilde S^{2}\,\partial_{\tilde S\tilde S}w\\
        &+\sup_{\nu\in\R}\Bigg(\left(\nu\, \tilde Z-\frac{\eta}{\kappa}\,\tilde Z^{3/2}\,\nu^{2}\right)\partial_{\tilde x}w-\nu\,\partial_{y}w-c\,\nu\,\partial_{\tilde Z}w-c\,\nu\,\partial_{\tilde S}w\Bigg)\ ,
    \end{split}
\end{equation}}
with terminal condition \begin{equation} \label{eq:ANX:termcondu}
w(T,\tilde x,y,\tilde Z,\tilde S)=\tilde x+\tilde y\,\tilde Z-\alpha\,\tilde y^{2}\,.
\end{equation}
The form of the terminal condition \eqref{eq:ANX:termcondu} suggests the ansatz
\begin{equation}
\label{eq:ANX:ansatz1}
    w(t,\tilde x,y,\tilde Z,\tilde S)=\tilde x+y\,\tilde Z+\theta(t,y,\tilde Z,\tilde S)\,,
\end{equation}
which we substitute into \eqref{eq:ANX:hjbu} to obtain {
\begin{equation}\label{eq:ANX:hjbtheta1}
    \begin{split}
        0\,=\,&\partial_{t}\theta-\phi\,y^{2}+\beta\,\left(\tilde S-\tilde Z\right)\,\left(y+\partial_{\tilde Z}\theta\right)+\frac{1}{2}\,\gamma^{2}\,\tilde Z^{2}\,\partial_{\tilde Z\tilde Z}\theta+\frac{1}{2}\,\sigma^{2}\,\tilde S^{2}\,\partial_{\tilde S\tilde S}\theta\\
        &+\sup_{\nu\in\R}\left(-\frac{\eta}{\kappa}\,\tilde Z^{3/2}\,\nu^{2}-\nu\,\left(\partial_{y}\theta+c\left(y+\partial_{\tilde Z}+\partial_{\tilde S}\theta\right)\right)\right)\,, 
    \end{split}
\end{equation}}
with terminal condition $\theta(T,y,\tilde Z,\tilde S)= -\tilde \alpha\,y^2\,.$

The first two terms on the right-hand side of \eqref{eq:ANX:ansatz1} are the marked-to-market value of the LT's holdings and the last term is the additional value that the LTs obtains by following the optimal strategy. Next, solve the first order condition in \eqref{eq:ANX:hjbtheta1} to obtain the optimal trading speed in feedback form{
\begin{equation}\label{eq:ANX:optimalspeed}
    \nu^{*}=-\frac{\kappa}{2\,\eta}\,\tilde Z^{-3/2}\,\left(\partial_{y}\theta+c\left(y+\partial_{\tilde Z}\theta+\partial_{\tilde S}\theta\right)\right)\,,
\end{equation}}
and substitute \eqref{eq:ANX:optimalspeed} into \eqref{eq:ANX:hjbtheta1} to write {
\begin{align}\label{eq:ANX:hjbtheta2}
    0=&\,\partial_{t}\theta-\phi\,y^{2}+\beta\,\left(\tilde S-\tilde Z\right)\,\left(y+\partial_{\tilde Z}\theta\right)+\frac{1}{2}\,\gamma^{2}\,\tilde Z^{2}\,\partial_{\tilde Z\tilde Z}\theta+\frac{1}{2}\,\sigma^{2}\,\tilde S^{2}\,\partial_{\tilde S\tilde S}\theta\\&+\frac{\kappa}{4\,\eta}\,\tilde Z^{-3/2}\,\left(\partial_{y}\theta+c\left(y+\partial_{\tilde Z}\theta+\partial_{\tilde S}\theta\right)\right)^{2}\,.
\end{align}}

The functional form of the convexity costs leads to {the semilinear PDE \eqref{eq:ANX:hjbtheta2}} which we cannot  solve in closed form. The optimal trading speed in feedback form is a function of the solution to the semilinear PDE, see \eqref{eq:ANX:optimalspeed}. In practice, one can use a numerical scheme to compute the optimal trading speed, which, in our case, is too computationally expensive for it to be deployed in real time by market participants.

\subsection{Optimal strategy with constant convexity costs \label{apx:standard model}}

Let $\zeta\ge 0$ denote the constant convexity costs. Here, the value function of the LT is given by
\begin{equation}\label{eq:valuefunc0}
    u^{\zeta}(t,x,y,Z,S)=\underset{\nu^{\zeta}\in\mathcal{A}^{\zeta}}{\sup}\,u^{\zeta}(t, x,y,Z,S)\,.
\end{equation}
The optimal trading strategy is
\begin{equation}
\label{eq:optimalspeed_feedback_j2}
\nu^{\zeta,\star}\left(t,{y},Z,S\right)=-\frac{1}{\eta\,\zeta}A_{\zeta}(t)\, y+\frac{1}{2\,\eta\,\zeta}B_{\zeta}(t)\,(S-Z)\ ,
\end{equation}
and the value function is 
\begin{align}
\label{eq:valuefunction_j2}
u^{\zeta}(t,{x},{y},Z,S) = &\,  x +  y \, Z + A_{\zeta}(t)\,{y}^{2}+B_{\zeta}(t)\,Z\,{y}+C_{\zeta}(t)\,{y}\,S+D_{\zeta}(t)\,{y}+ E_{\zeta}(t)\,Z^{2}
\\&+F_{\zeta}(t)\,S^{2} +G_{\zeta}(t)\,Z\,S\,,
\end{align}
where 
\begin{equation}
\label{eq:solution_system_ode_conve}
\begin{cases}
A_{\zeta}(t) & =\sqrt{\phi\,\eta\,\zeta}\,\tanh\left(\frac{\sqrt{\phi}}{\sqrt{\eta\,\zeta}}\,(T-t)+\arctanh\left(-\frac{\alpha}{\sqrt{\phi\,\eta\,\zeta}}\right)\right)\ ,\\
B_{\zeta}(t) & =-\int_{t}^{T}\beta\exp\left(-\int_{t}^{s}\left(\beta-\frac{1}{\eta\,\zeta}A_{\zeta}(u)\right)du\right)ds\, ,\\
C_{\zeta}(t) & =-B_{\zeta}(t)\, ,\\
E_{\zeta}(t) & =\int_{t}^{T}\exp\left(-(\gamma^{2}-2\,\beta)(t-s)\right)\frac{1}{4\,\eta\,\zeta}\,B_{\zeta}(s)^{2}\,ds\, ,\\
F_{\zeta}(t) & =\int_{t}^{T}\exp\left(-\sigma^{2}(t-s)\right)\,\left(\beta\, G_{\zeta}(s)+\frac{1}{4\,\eta\,\zeta}\,C_{\zeta}(s)^{2}\right)ds\, ,\\
G_{\zeta}(t) & =\int_{t}^{T}\exp\left(\beta\,(t-s)\right)\left(2\,\beta\, E_{\zeta}(s)-\frac{1}{2\,\eta\,\zeta}\,B_{\zeta}(s)^{2}\right)\,ds\,.
\end{cases}
\end{equation}

The first term on the right-hand side of \eqref{eq:optimalspeed_feedback_j2} is the optimal liquidation rate in the continuous Almgren--Chriss model. The second term is an arbitrage component; it accounts for the spread between the marginal rate $Z$ and the oracle rate $S$.

{
\section{Model of Section \ref{sec:Model2} \label{apx:Model2}}
\subsection{A semilinear PDE \label{apx:model2PDE}}

For each $(t,x,y,Z,\kappa )\in[0,T]\times\R\times\R\times\R_{++}\times\R_{++},$ and for each admissible control $\nu\in\mathcal{A}$ the performance criterion of the LT is given by
\begin{align}
\label{eq:perfcriteria_model2}
    u^{\nu}(t,{x},{y},Z,\kappa)=\E_{t,{x},{y},Z,\kappa }\left[{x}_{T}^{\nu}+{y}_{T}^{\nu}\,Z_T-\alpha\left({y}_{T}^{\nu}\right)^{2}-\phi\int_{t}^{T}\left({y}_{s}^{\nu}\right)^{2}ds\right]\ ,
\end{align}
and the value function is
\begin{align}
\label{eq:valuefunc_model2}
    u(t,{x},{y},Z,\kappa )=\underset{\nu\in\mathcal{A}}{\sup}\,u^{\nu}(t,{x},{y},Z,\kappa )\,.
\end{align}

The value function \eqref{eq:valuefunc_model2} is the unique classical solution to the HJB equation
\begin{equation}
    \label{eq:HJB_model2}
    \begin{split}
        0=&\,\partial_{t}w-\phi\,y^{2}+\frac{1}{2}\,\gamma^{2}\,Z^{2}\,\partial_{ZZ}w+\frac{1}{2}\,\varsigma^{2}\,\kappa^{2}\,\partial_{\kappa\kappa}w\\ 
        &+\underset{\nu\in\R}{\sup}\left(\left(\nu\,Z-\eta\,\frac{Z^{3/2}}{\kappa}\,\nu^{2}\right)\partial_{x}w-\nu\,\partial_{y}w\right)\,,
    \end{split}
\end{equation}
with terminal condition
\begin{equation}\label{eq:termcondu_model2}
    w(T,x,y,Z,\kappa) =x+y\,Z-\alpha\,y^{2}\,.
\end{equation}

The terminal condition \eqref{eq:termcondu_model2} suggests the ansatz
\begin{equation}
    \label{eq:ansatz1_Model2}
    w(t,x,y,Z,\kappa)=x+y\,Z+\theta(t,y,Z,\kappa)\,,
\end{equation}
which we justify by the following proposition, for which a proof is straightforward.

\begin{prop}
Assume there exists a function $\theta\in C^{1,1,2,2}([0,T]\times\R\times\R_{++}\times\R_{++})$ that solves
\begin{equation}\label{eq:PDE_Model2}
    0=\partial_{t}\theta-\phi\,y^{2}+\frac{1}{2}\,\gamma^{2}\,Z^{2}\,\partial_{ZZ}\theta+\frac{1}{2}\,\varsigma^{2}\,\kappa^{2}\,\partial_{\kappa\kappa}\theta+\underset{\nu\in\R}{\sup}\left(-\eta\,\frac{Z^{3/2}}{\kappa}\,\nu^{2}-\nu\,\partial_{y}\theta\right)\,,
\end{equation}
with terminal condition
\begin{equation}\label{eq:terminal_liquid_ansatz}
    \theta(T,y,Z,\kappa)=-\alpha\,y^{2}\,.
\end{equation}
Then, the function
\begin{equation}
\label{eq:ansatz_sec4}
    \begin{split}
        w(t,x,y,Z,\kappa)=x+y\,Z+\theta(t,y,Z,\kappa)
    \end{split}
\end{equation}
is a solution to \eqref{eq:HJB_model2} with terminal condition \eqref{eq:termcondu_model2}\,.
\end{prop}

Next, solve the first order condition in \eqref{eq:PDE_Model2} to obtain the LT's trading speed in feedback form
\begin{equation}\label{eq:strat_liquid}
    \nu^{\star}=	-\frac{\kappa}{2\,\eta}\,\partial_{y}\theta\,Z^{-3/2}\,.
\end{equation}
Substitute \eqref{eq:strat_liquid} into \eqref{eq:PDE_Model2} to write
\begin{equation}\label{eq:liquid_pde_theta}
    0=\partial_{t}\theta-\phi\,y^{2}+\frac{1}{2}\,\gamma^{2}\,Z^{2}\,\partial_{ZZ}\theta+\frac{1}{2}\,\varsigma^{2}\,\kappa^{2}\,\partial_{\kappa\kappa}\theta+\frac{\kappa}{4\,\eta}\,\left(\partial_{y}\theta\right)^{2}\,Z^{-3/2}\,.
\end{equation}

Finally, simplify \eqref{eq:liquid_pde_theta} with the ansatz
\begin{equation}\label{eq:second_liquid_ansatz}
    \theta(t,y,Z,\kappa)=\theta_{0}(t,Z,\kappa)+\theta_{1}(t,Z,\kappa)\,y+\theta_{2}(t,Z,\kappa)\,y^{2}\,,
\end{equation}
which is justified by the following proposition, for which a proof is straightforward.
\begin{prop}
Assume there exist functions $\theta_0\in C^{1,2,2}([0,T]\times\R_{++}\times\R_{++})$, $\theta_1\in C^{1,2,2}([0,T]\times\R_{++}\times\R_{++})$, and $\theta_2\in C^{1,2,2}([0,T]\times\R_{++}\times\R_{++})$ which solve the system of PDEs
\begin{equation}\label{eq:liquid_system}
    \left\{
    \begin{aligned}
    &0=\partial_{t}\theta_{2}-\phi+\frac{1}{2}\,\gamma^{2}\,Z^{2}\,\partial_{ZZ}\theta_{2}+\frac{1}{2}\,\varsigma^{2}\,\kappa^{2}\,\partial_{\kappa\kappa}\theta_{2}+\frac{\kappa}{\eta}\,\theta_{2}^{2}\,Z^{-3/2}\,,\\
    &0=\partial_{t}\theta_{1}+\frac{1}{2}\,\gamma^{2}\,Z^{2}\,\partial_{ZZ}\theta_{1}+\frac{1}{2}\,\varsigma^{2}\,\kappa^{2}\,\partial_{\kappa\kappa}\theta_{1}+\frac{\kappa}{\eta}\,\theta_{1}\,\theta_{2}\,Z^{-3/2}\,,\\
    &0=\partial_{t}\theta_{0}+\frac{1}{2}\,\gamma^{2}\,Z^{2}\,\partial_{ZZ}\theta_{0}+\frac{1}{2}\,\varsigma^{2}\,\kappa^{2}\,\partial_{\kappa\kappa}\theta_{0}+\frac{\kappa}{4\,\eta}\,\theta_{1}^{2}\,Z^{-3/2}\,,
    \end{aligned}
    \right.
\end{equation}
on $[0,T)\times\R_{++}\times\R_{++}$ with terminal conditions
\begin{equation}\label{eq:liquid_system_terminal}
    \theta_{2}(T,Z,\kappa)=-\alpha\,,\ \ \theta_{1}(T,Z,\kappa)=0\,,\ \,\text{and}\ \ \theta_{0}(T,Z,\kappa)=0\,.
\end{equation}
Then, the function 
\begin{equation}
    \theta(t,y,Z,\kappa)=\theta_{0}(t,Z,\kappa)+\theta_{1}(t,Z,\kappa)\,y+\theta_{2}(t,Z,\kappa)\,y^{2}\,,
\end{equation}
solves \eqref{eq:liquid_pde_theta} with terminal condition \eqref{eq:terminal_liquid_ansatz}\,.
\end{prop}

The optimal strategy in feedback form \eqref{eq:strat_liquid} is given by
\begin{equation}
\label{eq:feedbackformliquid}
    \nu^{\star}=-\frac{\kappa}{2\,\eta}\,\left(2\,\theta_2\,y+\theta_1\right)\,Z^{-3/2}\,.
\end{equation}

The system of PDEs in \eqref{eq:liquid_system} can be solved sequentially as follows. Solve the first PDE in the system to obtain $\theta_2$. Substitute $\theta_2$ is the second and third equations of the system so that the PDEs in $\theta_1$ and $\theta_0$ become linear. We cannot solve the semilinear PDE in $\theta_2$ in closed form, and providing an existence result is out of the scope of this work. 

\subsection{Closed-form approximation strategy \label{apx:closedformapprox model2}}

Here, we show how to obtain the closed-from approximation strategy \eqref{eq:pseudooptimal Model2} which can be implemented by the LT in real time and which accounts for the stochastic convexity costs in the pool.  First, similar to Section \ref{sec:Model}, we derive a strategy with constant convexity costs. 

Let $\zeta >0$ denote the constant convexity costs and let $\nu^{\zeta}\in\Ac^{\zeta}$.  Here, for each $\zeta$, the set $\mathcal A_t^{\zeta}$ of admissible strategies is similar to the admissible set \eqref{def:admissibleset_t_modelI}. Follow similar steps as those in Section \ref{sec:cfapprox} to obtain the new value function $$w^\zeta(t,x,y,\kappa,Z)=x+y\,Z+A_{\zeta}(t)\,y^{2}+B_{\zeta}(t)\,y+C_{\zeta}(t)\,,$$ 
where the system of PDEs in \eqref{eq:liquid_system} simplifies to the ODE system 
\begin{equation}\label{eq:liquid_system_new} 
\begin{cases}
0=A_{\zeta}'(t)-\phi+\frac{1}{\eta\,\zeta}A_{\zeta}(t)^{2}\,,\\
0=B_{\zeta}'(t)+\frac{1}{\eta\,\zeta}A_{\zeta}(t)\,B_{\zeta}(t)\,,\\
0=C_{\zeta}'(t)+\frac{1}{4\,\eta\,\zeta}B_{\zeta}(t)^{2}\,,
\end{cases} \implies \begin{cases}
A_{\zeta}(t)=\sqrt{\phi\,\eta\,\zeta}\,\tanh\left(\sqrt{\frac{\phi}{\eta\,\zeta}}\,\,(T-t)+\arctanh\left(-\frac{\alpha}{\sqrt{\phi\,\eta\,\zeta}}\right)\right)\,,\\
B_{\zeta}(t)=0\,,\\
C_{\zeta}(t)=0\,,
\end{cases}
\end{equation}
and the optimal strategy  \eqref{eq:feedbackformliquid} becomes
\begin{equation}\label{eq:feedbackformliquid_new} 
\nu^{\zeta,\star}=-y\,\sqrt{\frac{\phi}{\eta\,\zeta}}\,\tanh\left(\sqrt{\frac{\phi}{\eta\,\zeta}}\,\,(T-t)+\arctanh\left(-\frac{\alpha}{\sqrt{\phi\,\eta\,\zeta}}\right)\right)\,.
\end{equation}

Take hyperrectangles as in Section \ref{sec:cfapprox} to obtain the closed-from approximation strategy \eqref{eq:pseudooptimal Model2} when rates form in the DEX.} 

\section{Model of Section \ref{sec:Model3} \label{apx:Model3}}

This section obtains the trading strategy \eqref{eq:model3:approx_strategy} when rates form simultaneously in both the DEX and the CEX. Let $\zeta$ denote the constant convexity costs in the AMM. 
The performance criterion of the LT is
\begin{align}
\label{eq:apx:perfcriteria_model3}
    u_{\nu}^{\zeta}(t,x,y,\boldsymbol{P})=\E_{t,x,y,\boldsymbol{P}}\left[x_{T}^{\zeta}+y_{T}^{\zeta}\,\boldsymbol{\mathcal{X}}\,\boldsymbol{P}_{t}-\alpha\,\left(y_{T}^{\zeta}\right)^{2}-\phi\,\int_{t}^{T}\left(y_{s}^{\zeta}\right)^{2}\,ds\right]\,,
\end{align}
and the value function is
\begin{align}
\label{eq:apx:valuefunc_model3}
    u^{\zeta}(t,x,y,\boldsymbol{P})=\underset{\nu^{\zeta}\in\mathcal{A}^{\zeta}}{\sup}\,u_{\nu}^{\zeta}(t,x,y,\boldsymbol{P})\, .
\end{align}

The value function \eqref{eq:apx:valuefunc_model3} is the unique classical solution to the HJB equation
\begin{align}
\label{eq:apx:HJB_model3}
0=\,&\partial_{t}w^\zeta-\phi\,y^{2}+\left(\boldsymbol{\overline{P}}-\boldsymbol{P}\right)^{\intercal}\,\boldsymbol{{\Pi}}\,\partial_{\boldsymbol{P}}w^\zeta+\frac{1}{2}\,\text{Tr}\,\left(\boldsymbol{\Sigma}\,\partial_{\boldsymbol{P}\boldsymbol{P}}w^\zeta\right)\\&+\sup_{\nu\in\mathbb{R}}\left(\left(\boldsymbol{\mathcal{X}}\,\boldsymbol{P}\,\nu-\eta\,\zeta\,\nu^{2}\right)\partial_{x}w^\zeta-\nu\,\partial_{y}w^\zeta\right)\ ,
\end{align}
with terminal condition
\begin{equation}\label{eq:apx:termcondu_model3}
    w^\zeta(T,x,y,\boldsymbol{P})=x+y\,\boldsymbol{\mathcal{X}}\,\boldsymbol{P}-\alpha\,y^{2}\, .
\end{equation}

The terminal condition \eqref{eq:apx:termcondu_model3} suggests the ansatz 
\begin{equation}
    \label{eq:apx:ansatz1_Model3}
    w^\zeta(t,x,y,\boldsymbol{P})=x+y\,\boldsymbol{\mathcal{X}}\,\boldsymbol{P}+\theta^{\zeta}(t,y,\boldsymbol{P})\,,
\end{equation}
which leads to the  HJB  
\begin{equation}\label{eq:apx:PDE_Model3}
0=\partial_{t}\theta^{\zeta}-\phi\,y^{2}+\left(\boldsymbol{\overline{P}}-\boldsymbol{P}\right)^{\intercal}\,\boldsymbol{{\Pi}}\,\left(y\,\boldsymbol{\mathcal{X}}^{\intercal}+\partial_{\boldsymbol{P}}\theta^{\zeta}\right)+\frac{1}{2}\,\text{Tr}\,\left(\boldsymbol{\Sigma}\,\partial_{\boldsymbol{P}\boldsymbol{P}}\theta^{\zeta}\right)\ +\frac{\left(\partial_{y}\theta^{\zeta}\right)^{2}}{4\,\eta\,\zeta}\,,
\end{equation}
with terminal condition
\begin{equation}\label{eq:terminal_liquid_ansatz}
   \theta^{\zeta}(T,y,\boldsymbol{P})=-\alpha\,y^{2}\ ,
\end{equation}
where the optimal strategy in feedback form is 
\begin{equation}\label{eq:apx:strat feedback model 3}
    \nu^{\zeta, \star}=-\frac{1}{2\,\eta\,\zeta}\partial_{y}\theta^{\zeta}\,.
\end{equation}

Following similar steps to those in Sections \ref{apx:Model} and \ref{apx:Model2}, we use the ansatz 
\begin{equation}\label{eq:apx:ansatz 2 model 3}
\theta^{\zeta}(t,y,\boldsymbol{P})=\,A_{\zeta}(t)\,y^{2}+y\,B_{\zeta}(t)\,\boldsymbol{P}+C_{\zeta}(t)\,y+\boldsymbol{P}^{\intercal}D_{\zeta}(t)\,\boldsymbol{P}+E_{\zeta}(t)\,\boldsymbol{P}+F_{\zeta}(t)\,,    
\end{equation}
to obtain the system of ODEs:
\begin{equation}\label{eq:odesystem model 3}
\begin{cases}
0= & A'_\zeta(t)-\phi+\frac{1}{\eta\,\zeta}A_{\zeta}^{2}(t)\,,\\
0= & B_{\zeta}'(t)-\boldsymbol{\mathcal{X}}\,\boldsymbol{\Pi}^{\intercal}-B_{\zeta}(t)\,\boldsymbol{\Pi}^{\intercal}+\frac{1}{\eta\,\zeta}\,A_{\zeta}(t)\,B_{\zeta}(t)\,,\\
0= & C_{\zeta}'(t)+\boldsymbol{\overline{P}}^{\intercal}\,\boldsymbol{\beta}\,\boldsymbol{\mathcal{X}}^{\intercal}+\boldsymbol{\overline{P}}^{\intercal}\,\boldsymbol{\Pi}\,B_{\zeta}(t)^{\intercal}+\frac{1}{\eta\,\zeta}\,A_{\boldsymbol{\zeta}}(t)\,C_{\zeta}(t)\,,\\
0= & D_{\zeta}'(t)+\frac{1}{4\,\eta\,\zeta}B_{\zeta}(t)^{\intercal}\,B_{\zeta}(t)\,,\\
0= & E_{\zeta}'(t)-E_{\zeta}(t)^{\intercal}\,\boldsymbol{{\Pi}}^{\intercal}+\frac{1}{2\,\eta\,\zeta}\,C_{\zeta}(t)^{\intercal}\,B_{\zeta}(t)\,,\\
0= & F_{\zeta}'(t)+\boldsymbol{\overline{P}}^{\intercal}\,\boldsymbol{{\Pi}}\,E_{\zeta}(t)+Tr\left(\boldsymbol{\Sigma}\,D_{\zeta}(t)\right)+\frac{1}{4\,\eta\,\zeta}\,C_{\zeta}(t)^{\intercal}\,C_{\zeta}(t)\,,
\end{cases}
\end{equation}
with terminal conditions $A_{\zeta}(T)=-\alpha$ and $B_{\zeta}(t)=C_{\zeta}(t)=D_{\zeta}(t)=E_{\zeta}(t)=F_{\zeta}(t)=0\,.$

The system of ODEs \eqref{eq:odesystem model 3} admits the solution 
\begin{equation}\label{eq:apx:solution ode system model 3}
\begin{cases}
A_{\zeta}(t)= & \sqrt{\phi\,\eta\,\zeta}\,\tanh\left(\sqrt{\frac{\phi}{\eta\,\zeta}}\,(T-t)+\arctanh\left(-\frac{\alpha}{\sqrt{\phi\,\eta\,\zeta}}\right)\right)\,,\\
B_{\zeta}(t)= & -\int_{t}^{T}e^{\int_{t}^{s}\frac{1}{\eta\,\zeta}\,A_{\boldsymbol{\zeta}}(u)\,du}\,\boldsymbol{\mathcal{X}}\,\boldsymbol{\Pi}^{\intercal}\,e^{-\boldsymbol{{\Pi}}^{\intercal}(s-t)}ds\,,\\
C_{\zeta}(t)= & -\boldsymbol{\overline{P}}\,B_{\zeta}(t)\,,\\
D_{\zeta}(t)= & E_{\zeta}(t)=F_{\zeta}(t)=0\,,
\end{cases}
\end{equation}
so the optimal strategy in \eqref{eq:apx:strat feedback model 3} becomes
$$\nu^{\zeta,\star}=-\frac{1}{\eta\,\zeta}\,A_{\zeta}(t)\,y+\frac{1}{2\,\eta\,\zeta}\,B_{\zeta}(t)\,\left(\boldsymbol{\overline{P}}-\boldsymbol{P}\right)\,.$$

Finally, take hyperrectangles as in Section \ref{sec:cfapprox}
to obtain the closed-form approximation strategy \eqref{eq:model3:approx_strategy} when rates form in both the CEX and DEX.

}

\clearpage
\bibliographystyle{elsarticle-harv}
\bibliography{references}

\end{document}